\documentclass{article}

\usepackage[margin=1.2in]{geometry}
\usepackage[utf8]{inputenc} 
\usepackage[T1]{fontenc}    

\usepackage[numbers]{natbib}
\usepackage{tocloft}

\usepackage{amsmath,amsfonts,amsthm,amssymb,comment}

\usepackage{caption}
\usepackage{subcaption}
\usepackage{hyperref}       

\usepackage{url}            
\usepackage{booktabs}       
\usepackage{enumitem}

\usepackage{leftindex}

\usepackage{nicefrac}       
\usepackage{microtype}      
\usepackage{mathtools}
\usepackage{xfrac}
\usepackage{csquotes}

\usepackage{bbm}
\usepackage[]{xcolor}

\usepackage{cleveref}

\usepackage{graphicx}
\usepackage{enumitem}
\usepackage{autonum}
\usepackage{multirow}

\newcommand{\EE}{\mathbb{E}}

\newcommand{\Fcal}{\mathcal{F}}

\newcommand\independent{\protect\mathpalette{\protect\independenT}{\perp}}\def\independenT#1#2{\mathrel{\rlap{$#1#2$}\mkern2mu{#1#2}}}

\newcommand{\tnull}{\textnormal{null}}
\newcommand{\nd}{\textnormal{d}}
\newcommand{\Rcal}{\mathcal{R}}
\newcommand{\Tcal}{\mathcal{T}}

\newcommand{\rejset}{R}

\newcommand{\nullSimes}{\textnormal{null-Simes}}
\newcommand{\BH}{\textnormal{BH}}
\newcommand{\oBH}{\textnormal{oBH}}
\newcommand{\oSBH}{\textnormal{oSBH}}
\newcommand{\oeBH}{\textnormal{oe-BH}}
\newcommand{\LORD}{\textnormal{LORD}}
\newcommand{\SAFFRON}{\textnormal{SAFFRON}}
\newcommand{\FDR}{\textnormal{FDR}}
\newcommand{\FDP}{\textnormal{FDP}}
\newcommand{\supFDR}{\textnormal{SupFDR}}
\newcommand{\StopFDR}{\textnormal{StopFDR}}

\newcommand{\prob}[1]{\mathbb{P}\left(#1\right)}
\newcommand{\expect}{\mathbb{E}}

\newcommand{\ind}[1]{\mathbbm{1}\left\{#1\right\}}

\newcommand{\naturals}{\mathbb{N}}
\newcommand{\reals}{\mathbb{R}}

\graphicspath{{fig/}{fig/plots/}}

\newtheorem{theorem}{Theorem}[section]
\newtheorem{corol}[theorem]{Corollary}
\newtheorem{proposition}[theorem]{Proposition}
\newtheorem{lemma}[theorem]{Lemma}
\theoremstyle{remark}

\newtheorem{remark}{Remark}
\newtheorem{example}{Example}
\newtheorem{definition}{Definition}
\newtheorem{fact}{Fact}

\newif\ifarxiv

\arxivtrue

\title{An online generalization of the \\ (e-)Benjamini-Hochberg procedure}

\author{Lasse Fischer\thanks{University of Bremen, Germany. Email: \texttt{fischer1@uni-bremen.de}} \and Ziyu Xu\thanks{Carnegie Mellon University, USA. Email: \texttt{xzy@cmu.edu} } \and Aaditya Ramdas\thanks{Carnegie Mellon University, USA. Email: \texttt{aramdas@cmu.edu} }}

\begin{document}

\maketitle
\begin{abstract}
In online multiple testing, the hypotheses arrive one by one, and at each time we must immediately reject or accept the current hypothesis solely based on the data and hypotheses observed so far. Many online procedures have been proposed, but none of them are generalizations of the Benjamini-Hochberg (BH) procedure based on p-values, or of the e-BH procedure that uses e-values. In this paper, we consider a relaxed problem setup that allows the current hypothesis to be rejected at \emph{any later step}.
We show that this relaxation allows us to define --- what we justify extensively to be --- the natural and appropriate online extension of the BH and e-BH procedures.
We show that the FDR guarantees for BH (resp. e-BH) and online BH (resp. online e-BH) are identical under positive, negative or arbitrary dependence, at fixed and stopping times. Further, the online BH (resp. online e-BH) rule recovers the BH (resp. e-BH) rule as a special case when the number of hypotheses is known to be fixed. Of independent interest, our proof techniques also allow us to prove that numerous existing online procedures, which were known to control the FDR at fixed times, also control the FDR at stopping times. 
\end{abstract}

\newpage
\tableofcontents
\newpage

\section{Introduction}
Suppose we observe a stream of (null) hypotheses $(H_t)_{t\in \mathbb{N}}$ that comes either with a sequence of p-values $(P_t)_{t\in \mathbb{N}}$ or with a sequence of e-values $(E_t)_{t\in \mathbb{N}}$. Let
$\mathbb{P}$ denote the true but unknown probability distribution and $\mathbb{E}$  the corresponding expected value, while $\mathbb{P}_{H_0}$ and $\mathbb{E}_{H_0}$ denote the probability and expectation under the null hypothesis (that is, according to any distribution that satisfies $H_0$). A p-value $P_t$ for $H_t$ is a random variable with values in $[0,1]$ such that $\mathbb{P}_{H_t}(P_t\leq x)\leq x$ for all $x\in [0,1]$. In contrast, an e-value $E_t$ is a nonnegative random variable with $\mathbb{E}_{H_t}[E_t]\leq 1$. Previous works \citep{foster2008alpha, javanmard2018online, fischer2024online, xu2024online} defined an online multiple testing procedure as a sequence of test decisions $(r_t)_{t\in \mathbb{N}}$, where each $r_t$ is only allowed to depend on the p-values (or e-values) $P_1,\ldots, P_t$. In this paper we relax this setup, which allows to define more powerful procedures that are valid under more general stopping rules. This allows us to propose (what we justify to be) online generalizations of the Benjamini-Hochberg (BH) and e-BH procedures (recapped later).

\subsection{Online ARC (accept-to-reject) procedures} An \textit{online ARC procedure} is an online multiple testing method that can change earlier accept decisions (or non-rejections) into a rejection at a later point in time. The motivation is that early in the procedure, there have been no discoveries yet, and so the level at which hypotheses are tested at must be quite stringent leading to fewer discoveries initially. However, as the number of discoveries builds up, we may realize that we could have tested earlier hypotheses at more relaxed levels, leading to rejections at earlier times.

In more detail, an online ARC procedure goes through all p-values $P_1, P_2,\ldots$ (or e-values) and at each time $t\in \mathbb{N}$ it decides based on $P_1,\ldots,P_t$ which of the hypotheses $H_1,\ldots, H_t$ are to be rejected with the restriction that as soon as a hypothesis is rejected, this decision cannot be overturned in the future and thus the rejection remains during the entire testing process. However, in contrast to a fully online multiple testing procedure, an online ARC procedure allows to turn an accepted hypothesis into a rejected one based on information obtained in the future. This is why we call this online with \emph{accept-to-reject changes} (online ARC) procedure. Formally one can define an online ARC procedure as follows.

\begin{definition}
    An online ARC procedure is a nested sequence of rejection sets $R_1\subseteq R_2 \subseteq R_3 \subseteq \ldots $, where each $R_t\subseteq [t] \coloneqq \{1,\ldots,t\}$ only depends on the p-values or e-values that are known up to step $t$.
\end{definition}


In statistical hypothesis testing, the null hypothesis $H_t$ usually states the status quo which is challenged by testing it. If $H_t$ is rejected, this can be seen as evidence that $H_t$ is not true. This evidence is used to draw conclusions and derive recommendations for the future. However, if $H_t$ is accepted, this is typically not seen as evidence that $H_t$ is true, but just as insufficient evidence to reject $H_t$. Therefore, actions are typically not changed based on accepted hypotheses. For these reasons, we believe that reversing rejections should be forbidden,  while accept-to-reject changes are useful in many applications.

For example, an important online multiple testing application is A/B testing in tech companies \citep{kohavi2013online, ramdas2017online}. Here, a rejection might lead to an adjustment of a website or an app. However, after an acceptance there won't be any changes. Consequently, reversing rejections can lead to undoing unnecessary adjustments made to a website, which can cost money and time, whereas turning accepted hypotheses into rejected ones may lead to additional improvements of the website at a later point. The same logic applies to many other online multiple testing tasks such as platform trials \citep{robertson2023platform, zehetmayer2022online} and genomic studies \citep{aharoni2014generalized, 10002015global}.

Online multiple testing procedures are a special case of online ARC procedures, where we disallow all changes to decisions. Sometimes, online multiple testing procedures are suitable for the applications, leading to a rich literature on the topic, but ARC procedures also have their place, and this particular paper focuses on designing online ARC procedures.

Interestingly, offline procedures can also be interpreted as special online ARC procedures where we do not make any rejections until the final step.

\subsection{Problem Setup}
Let $I_0\subseteq \mathbb{N}$ be the index set of true hypotheses. The false discovery proportion (FDP) of a rejection set $R$ is defined as
$$
\text{FDP}(R)\coloneqq\frac{|R\cap I_0|}{|R| \lor 1}.
$$
The false discovery rate (FDR) is the expected FDP: $\text{FDR}(R)=\mathbb{E}\left[ \text{FDP}(R)\right]$. We let FDP$_t$ and FDR$_t$ be shorthand for FDP$(R_t)$ and FDR$(R_t)$, respectively, for each $t \in \naturals$.
The type FDR control usually considered in online multiple testing is defined as:
$$
  \textnormal{OnlineFDR} \coloneqq \sup_{t \in \naturals}\ \FDR_t.
$$
\citet{javanmard2018online} provided an algorithm that controlled OnlineFDR below a fixed level $\alpha \in [0, 1]$ for independent p-values. However, this does not guarantee FDR control at any data-adaptive stopping time (like the time of the hundredth rejection); this is severe restriction in the online setting. For this reason, \citet{xu2022dynamic} introduced the SupFDR, defined as
$$
\text{SupFDR}\coloneqq\mathbb{E}\left[ \sup_{t\in \mathbb{N}}\  \text{FDP}_t \right].
$$

We can also define a notion of FDR for arbitrary (data-dependent) stopping times
\begin{align}
    \textnormal{StopFDR} \coloneqq \sup_{\tau \in \Tcal}\ \FDR_\tau,
\end{align} 
where $\Tcal$ denotes the set of all stopping times.
In this paper, we will generally use $\alpha \in (0, 1]$ to denote a constant used to induce a fixed level of error control unless otherwise specified. We note that enforcing control of SupFDR and StopFDR have the following relationship,
\begin{align}
&\text{SupFDR} \leq \alpha \implies \textnormal{StopFDR} \leq \alpha \\ &\text{StopFDR} \leq \alpha \implies \textnormal{SupFDR} \leq \alpha(1 + \log(\alpha^{-1})).
\end{align}
Controlling the SupFDR ensures valid FDR control at arbitrary (data-dependent) stopping times since the supremum of $\FDP_t$ will exceed $\FDP_\tau$ at any single time $\tau$. Interestingly, a reverse relationship also holds, meaning control of the FDR at stopping times implies control of the SupFDR, albeit at an inflated level (we prove this result in Appendix~\ref{sec:appn_supFDR_FDR}).

In addition to the aforementioned false discovery metrics that involve all time steps, we will also consider the SupFDR until some fixed time $K$, defined as
$$
\text{SupFDR}^K\coloneqq\mathbb{E}\left[ \sup_{1\leq t\leq K} \text{FDP}_t \right].
$$
The $\text{SupFDR}^K$ can be useful if there is a fixed maximum number of hypotheses $K$, but we may want to stop early.

\subsection{Our contributions}
In this paper we introduce the online e-BH and online BH procedure, online ARC versions of the e-BH and BH procedures \citep{wang2022false, benjamini1995controlling}.

The e-BH procedure was recently proposed by \citet{wang2022false} and provides an e-value analog of the seminal p-value based Benjamini-Hochberg (BH) procedure \citep{benjamini1995controlling}. We show that online e-BH controls the SupFDR under arbitrary dependence between the e-values (Section~\ref{sec:supFDR_eBH}). To the best of our knowledge, this is the first nontrivial online procedure with SupFDR control under arbitrary dependence. In addition, it uniformly improves e-LOND, the current state-of-the-art online procedure for FDR control with arbitrarily dependent e-values.
Next, we demonstrate how the ``boosting'' approach for the \emph{offline} e-BH procedure \citep{wang2022false} can be extended to the \emph{online} e-BH procedure (Section~\ref{sec:boosting}). We also propose a new boosting technique for locally dependent e-values which improves on boosting under arbitrary dependence. A comparison of the guarantees provided by the e-BH and online e-BH procedure is given in Table~\ref{tab:e-BH}.

\ifarxiv
\begin{table*}[h!]
    \centering
    \resizebox{\textwidth}{!}{%
    \begin{tabular}{l|l|l}
       Assumption & e-BH & Online e-BH \\ \hline
        Independence or PRDS (boosted) & $\text{FDR}\leq \alpha$ \citep{wang2022false} & $\text{OnlineFDR}\leq \alpha$ (Sec.~\ref{sec:boosting}) \\
        Arbitrary dependence & $\text{FDR}\leq \alpha$ \citep{wang2022false} & $\text{SupFDR}\leq \alpha$ (Sec.~\ref{sec:supFDR_eBH}) \\
        Local dependence (boosted) & Not considered before& $\text{SupFDR}\leq \alpha$ (Sec.~\ref{sec:boosting}) \\
         \hline
    \end{tabular}}
    \caption{Comparison of e-BH and online e-BH guarantees.}
    \label{tab:e-BH}
\end{table*}
\else
\begin{table*}[h!]
    \centering
    \resizebox{14cm}{!}{%
    \begin{tabular}{l|l|l}
       Assumption & e-BH & Online e-BH \\ \hline
        Independence or PRDS (boosted) & $\text{FDR}\leq \alpha$ \citep{wang2022false} & $\text{OnlineFDR}\leq \alpha$ (Sec.~\ref{sec:boosting}) \\
        Arbitrary dependence & $\text{FDR}\leq \alpha$ \citep{wang2022false} & $\text{SupFDR}\leq \alpha$ (Sec.~\ref{sec:supFDR_eBH}) \\
        Local dependence (boosted) & Not considered before& $\text{SupFDR}\leq \alpha$ (Sec.~\ref{sec:boosting}) \\
         \hline
    \end{tabular}}
    \caption{Comparison of e-BH and online e-BH guarantees.}
    \label{tab:e-BH}
\end{table*}
\fi

In contrast, we show that the online BH procedure controls the OnlineFDR under positive regression dependence on a subset (PRDS)  (Section~\ref{sec:ofdr-prds}). Furthermore, it controls the SupFDR at slightly inflated levels under weak positive (PRDN) and negative (WNDN) dependence assumptions (Sections~\ref{sec:supFDR_PRDN} and~\ref{sec:neg_dep}).  Under arbitrary dependence, online BH provides $\text{SupFDR}^K$ control at an inflated level if a maximum number of hypotheses $K$ is specified (\Cref{sec:supFDR_oBH_arb}). The guarantees of the BH and online BH procedure are summarized in Table~\ref{tab:BH}.



\ifarxiv
\begin{table*}[h!]
    \centering
    \resizebox{\textwidth}{!}{%
    \begin{tabular}{l|l|l}
       Assumption & BH & Online BH \\ \hline
        Independence & $\text{FDR}\leq \alpha$ \citep{benjamini1995controlling} & $\text{OnlineFDR}\leq \alpha$ (Sec.~\ref{sec:ofdr-prds}) \\
        PRDS & $\text{FDR}\leq \alpha$ \citep{ benjamini2001control} & $\text{OnlineFDR}\leq \alpha$ (Sec.~\ref{sec:ofdr-prds}) \\
        PRDN & $\text{FDR}\leq \alpha(1+\log(\alpha^{-1}))$ \citep{su_fdr-linking_theorem_2018} & $\text{SupFDR}\leq \alpha(1+\log(\alpha^{-1}))$ (Sec.~\ref{sec:supFDR_PRDN}) \\
        WNDN & $\text{FDR}\leq \alpha(3.18+\log(\alpha^{-1})) $ \citep{chi_multiple_testing_2024} & $\text{SupFDR}\leq \alpha(3.18+\log(\alpha^{-1})) $  (Sec.~\ref{sec:neg_dep}) \\
        Arbitrary dependence & $\text{FDR}\leq \alpha \ell_K $ \citep{benjamini2001control} & $\text{SupFDR}^K\leq \alpha \ell_K $  (Sec.~\ref{sec:supFDR_oBH_arb}) \\ \hline
    \end{tabular}}
    \caption{Comparison of BH and online BH guarantees.}
    \label{tab:BH}
\end{table*}
\else
\begin{table*}[h!]
    \centering
    \resizebox{15cm}{!}{%
    \begin{tabular}{l|l|l}
       Assumption & BH & Online BH \\ \hline
        Independence & $\text{FDR}\leq \alpha$ \citep{benjamini1995controlling} & $\text{OnlineFDR}\leq \alpha$ (Sec.~\ref{sec:ofdr-prds}) \\
        PRDS & $\text{FDR}\leq \alpha$ \citep{ benjamini2001control} & $\text{OnlineFDR}\leq \alpha$ (Sec.~\ref{sec:ofdr-prds}) \\
        PRDN & $\text{FDR}\leq \alpha(1+\log(\alpha^{-1}))$ \citep{su_fdr-linking_theorem_2018} & $\text{SupFDR}\leq \alpha(1+\log(\alpha^{-1}))$ (Sec.~\ref{sec:supFDR_PRDN}) \\
        WNDN & $\text{FDR}\leq \alpha(3.18+\log(\alpha^{-1})) $ \citep{chi_multiple_testing_2024} & $\text{SupFDR}\leq \alpha(3.18+\log(\alpha^{-1})) $  (Sec.~\ref{sec:neg_dep}) \\
        Arbitrary dependence & $\text{FDR}\leq \alpha \ell_K $ \citep{benjamini2001control} & $\text{SupFDR}^K\leq \alpha \ell_K $  (Sec.~\ref{sec:supFDR_oBH_arb}) \\ \hline
    \end{tabular}}
    \caption{Comparison of BH and online BH guarantees.}
    \label{tab:BH}
\end{table*}
\fi

As shown in the Tables~\ref{tab:e-BH} and~\ref{tab:BH}, both online procedures, online e-BH and online BH, provide OnlineFDR control whenever their offline counterparts provide FDR control and in some cases even control the SupFDR where the offline methods were only known to control the FDR. In addition, we will show that the \emph{online} versions are natural \emph{generalizations} of the \emph{offline} procedures, meaning e-BH and BH can be obtained by online e-BH and BH for a specific choice of weighting parameters (see Sections~\ref{sec:online_eBH} and~\ref{sec:onlineBH}). Therefore, we believe the names \emph{online e-BH} and \emph{online BH} are justified for our procedures.

Along the way, we also prove SupFDR control for all other existing online procedures with FDR control under arbitrary dependence. This includes the reshaped LOND procedure \citep{zrnic2021asynchronous} and e-LOND \citep{xu2024online}. Previously, these methods were only known to control the FDR at fixed times. In addition, we prove SupFDR control of LOND \citep{fisher2022online} under PRDN and WNDN at slightly inflated levels. We also provide new guarantees for the TOAD algorithm by \citet{fisher2022online}. 


\subsection{The TOAD algorithm}
While ARC procedures have not been explicitly named and conceptualized before, the only online ARC procedure that we are aware of (which is not already a fully online procedure) is the TOAD algorithm by \citet{fisher2022online}. The TOAD algorithm was introduced in the context of p-values and \enquote{decision deadlines} where one must make a decision about $H_t$ by a deadline $d_t$ based on a p-value $P_t$. One can view online procedures as having a deadline for making a decision as setting $d_t = t$, and online ARC procedures as setting $d_t = \infty$. In Fisher's framework, rejection-to-acceptance changes before the deadline were not explicitly disallowed, although not used by the procedures. In the Appendix~\ref{sec:e-TOAD} we present the e-TOAD algorithm, an e-value analog of the TOAD procedure with decision deadlines. Then, our online e-BH procedure can be viewed as an instance of our e-TOAD framework with unbounded deadlines, where e-TOAD extends TOAD and its guarantees from p-values to e-values. Further, while \citet{fisher2022online} showed that TOAD for arbitrarily dependent p-values controls the OnlineFDR, we extend his result to show that it actually controls the SupFDR. We also improve the power of TOAD under arbitrary dependence slightly by allowing individual shape functions for the hypotheses. Furthermore, we prove that TOAD for positively dependent p-values controls the SupFDR under PRDN and WNDN at a slightly inflated level.

\section{The online e-BH  procedure\label{sec:online_eBH}}
Let us first review the e-BH procedure for offline multiple testing. Let $K$ e-values $E_1,\ldots, E_K$ for the hypotheses $H_1,\ldots,H_K$ be given. The (base) e-BH procedure \citep{wang2022false} rejects hypotheses with the $k^*$ largest e-values, where
\begin{align}
    k^*=\max\left\{k\in \{1,\ldots,K\}: \sum_{j=1}^K \mathbbm{1}\{E_j\geq K/(k \alpha ) \} \geq k\right\},
\end{align}
with the convention $\max(\emptyset)=0$.
 For the online e-BH procedure we weight each hypothesis $H_i$ with a nonnegative constant $\gamma_i$ such that $\sum_{i\in \mathbb{N}} \gamma_i \leq 1$. Let $\pi_0 \coloneqq \sum_{i \in I_0} \gamma_i$ be the sum of the null weights. The online e-BH procedure is defined by the rejection sets
 \begin{align}
 &R_t^{\oeBH}\coloneqq\left\{i\leq t: E_i\geq \frac1{k_t^* \alpha \gamma_i}\right\}, t\in \mathbb{N}, \text{ where}
 \\
&k_t^* =\max\left\{k\in \{1,\ldots,t\}: \sum_{j=1}^t \mathbbm{1}\{E_j\geq 1/(k \alpha  \gamma_j) \} \geq k\right\},
\end{align}
with the conventions $\max(\emptyset)=0$ and $1/0=\infty$.
Note that $k_t^*$ is nondecreasing in $t$, which implies that $(R_t^{\oeBH})_{t\in \mathbb{N}}$ is a nested sequence of rejection sets such that the online e-BH procedure is indeed an online ARC procedure. Also note that the naive procedure that sets $R_t$ to be the rejection set of the offline e-BH procedure applied to $E_1,...,E_t$ is not an online ARC procedure: it may change its mind in both directions (possibly more than once).

If we choose $\gamma_i=1/K$ for $i=1,\ldots,K$ and $\gamma_i=0$ otherwise, the online e-BH becomes the e-BH procedure, and in that sense it is an \emph{online generalization of the e-BH procedure}.
In particular, the online e-BH allows to choose different weights for the hypotheses and also to test an infinite number of hypotheses. \citet{wang2022false} have already formulated a weighted e-BH procedure which is similar to this method; however, their approach does not immediately apply to an infinite number of hypotheses while we claim that our approach defines an online ARC procedure with SupFDR control.

\subsection{SupFDR control under arbitrary dependence\label{sec:supFDR_eBH}}

To show SupFDR control, we first define the following class of procedures.

\begin{definition}
A discovery set of finite cardinality, $\rejset \subseteq \naturals$, is considered \emph{self-consistent} if it satisfies $E_t \geq (\alpha \gamma_t |\rejset|)^{-1}$ for e-values (or $P_t \leq \alpha \gamma_t |\rejset|$ for p-values) for each $t \in \rejset$. We denote by $\Rcal(\alpha)$ the set of discovery sets that are self-consistent at level $\alpha$.
\end{definition}
The above definition is consistent with the definition of self-consistency as defined by \citet{blanchard2008two} and generalizes the definition (referred to as ``compliance'') introduced in \citet{su_fdr-linking_theorem_2018}. \citet{wang2022false} showed that self-consistent e-value procedures control FDR under arbitrary dependence in the offline multiple testing setting with uniform weights. We now show an extension of this result for general weights $(\gamma_t)_{t \in \naturals}$ and any self-consistent rejection set in an infinite stream of hypotheses.
\begin{proposition}
    \label{prop:e-self-consist}
    For a stream of infinite hypotheses and its corresponding arbitrarily dependent e-values $(E_t)_{t \in \naturals}$, $\expect\left[\sup_{R \in \Rcal(\alpha)} \FDP(R)\right] \leq \alpha$.
\end{proposition}
\begin{proof}
    We first prove it for the case where we have $K$ hypotheses.  Since $\Rcal(\alpha)$ is finite, there exists a $R^*$ that maximizes the FDP. Thus we have that
    \begin{align}
    \expect\left[\sup_{R \in \Rcal(\alpha)} \FDP(R)\right] &= \FDR(R^*) \\ &= \sum\limits_{i \in I_0} \expect\left[\frac{\ind{E_i \geq 1/(\alpha \gamma_i |R^*|)}}{|R^*| \vee 1}\right] \\ &\leq \alpha \cdot \sum\limits_{i \in I_0} \gamma_i \expect[E_i] \leq \pi_0\alpha \leq \alpha.
    \label{eq:finite-K-supfdr}
    \end{align}
Where the first inequality follows because $\mathbbm{1}\{x \geq 1\} \leq x$ for all $x \geq 0$, the second follows from $E_i$ being an e-value for each $i \in I_0$, and the last is due to $\pi_0 \leq 1$.

When we have an infinite stream of hypotheses, we know the following is true
\begin{align}
     \sup_{\rejset \in \Rcal(\alpha)} \FDP(R) = \lim_{t\rightarrow \infty} \sup_{\rejset \in \Rcal(\alpha), R\subseteq [t]} \FDP(R).
\end{align} Thus, by taking an expectation over all the terms in the above equation, we get
\ifarxiv
\begin{align}
     \expect\left[\sup_{\rejset \in \Rcal(\alpha)}\ \FDP(R) \right]
     = \lim_{t\rightarrow \infty} \expect\left[\sup_{\rejset \in \Rcal(\alpha), R\subseteq [t]}\ \FDP(R) \right]
     \leq \pi_0  \alpha \leq \alpha,
\end{align}
\else
\begin{align}
     \expect\left[\sup_{\rejset \in \Rcal(\alpha)}\ \FDP(R) \right]
     &= \lim_{t\rightarrow \infty} \expect\left[\sup_{\rejset \in \Rcal(\alpha), R\subseteq [t]}\ \FDP(R) \right]
     \\ &\leq \pi_0  \alpha \\ &\leq \alpha,
\end{align}
\fi
where the equality is by the Fatou-Lebesgue theorem and the last two inequalities are by \eqref{eq:finite-K-supfdr}.
\end{proof}

This immediately implies SupFDR control for online e-BH, since the rejection sets $R_t^{\oeBH}$ are self-consistent for all $t$.

\begin{theorem}\label{theo:SupFDR_control}
    The online e-BH procedure controls the SupFDR at level $\alpha$ under arbitrary dependence between the e-values.
\end{theorem}

The existing state-of-the-art method for OnlineFDR control under arbitrary dependence of the e-values is e-LOND \citep{xu2024online}. In the following subsection, we show that e-LOND even provides SupFDR control and that online e-BH uniformly improves e-LOND.

\subsection{Online e-BH uniformly improves e-LOND\label{sec:e-LOND}}
 The e-LOND algorithm \citep{xu2024online} rejects hypothesis $H_t$, if $E_t\geq (\alpha_t^{\text{e-LOND}})^{-1}$, where
 \ifarxiv
$$\alpha_t^{\text{e-LOND}}\coloneqq\alpha \gamma_t \left(|R_{t-1}^{\text{e-LOND}}|+1\right) \quad \text{and} \quad R_{t}^{\text{e-LOND}}\coloneqq\{i\leq t: E_i\geq (\alpha_i^{\text{e-LOND}})^{-1}\}. $$
\else
\begin{align}&\alpha_t^{\text{e-LOND}}\coloneqq\alpha \gamma_t \left(|R_{t-1}^{\text{e-LOND}}|+1\right) \quad \\ \text{and} \quad &R_{t}^{\text{e-LOND}}\coloneqq\{i\leq t: E_i\geq (\alpha_i^{\text{e-LOND}})^{-1}\}. \end{align}
\fi
 It follows immediately that $R_{t}^{\text{e-LOND}} \subseteq R_t^{\text{e-BH}}$ for all $t\in \mathbb{N}$, where $R_t^{\text{e-BH}}$ is the rejection set of the online e-BH procedure at step $t$. Hence, online e-BH uniformly improves e-LOND, meaning that the former makes at least the same discoveries as the latter on every run. \citet{xu2024online} only proved FDR control for e-LOND at fixed times.
 We note that e-LOND is also a self-consistent procedure so its SupFDR control follows immediately from Proposition~\ref{prop:e-self-consist}.
 \begin{proposition}\label{prop:SupFDR_e-lond}
    The e-LOND procedure \citep{xu2024online} controls SupFDR at level $\alpha$ under arbitrary dependence between the e-values.
\end{proposition}

The power gain of the online e-BH method compared to e-LOND is illustrated in Figure~\ref{fig:sim_e_lond}. The left plot shows that the online e-BH method substantially improves e-LOND when the signal of the alternative hypotheses is weak. This improvement is slightly smaller when the signal strong (right plot). The simulation setup is in Appendix~\ref{sec:sim_setup}.

\ifarxiv
\begin{figure*}[h!]
\centering
\includegraphics[width=0.8\textwidth]{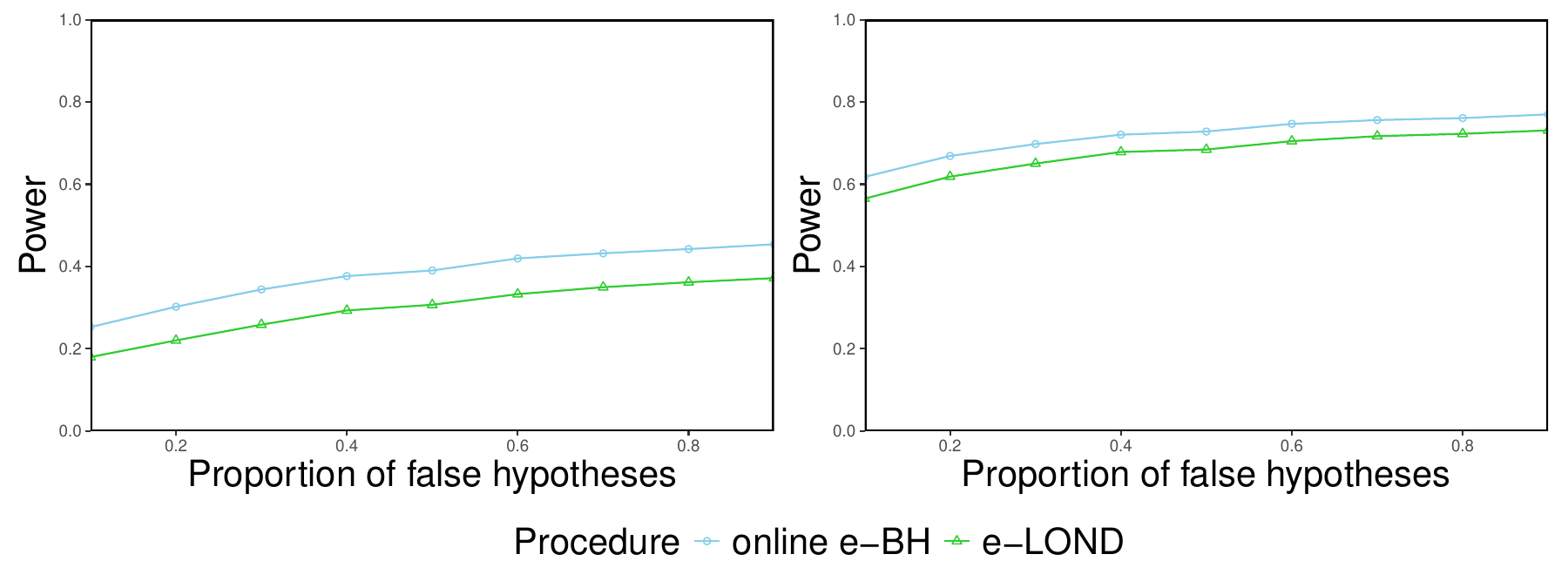}
\caption{Power comparison of online e-BH and e-LOND for different proportions of false hypotheses. In the left (right) plot, the signal of the alternative is weak (strong). 
\label{fig:sim_e_lond} }\end{figure*}
\else
\begin{figure*}[h!]
\centering
\includegraphics[width=14cm]{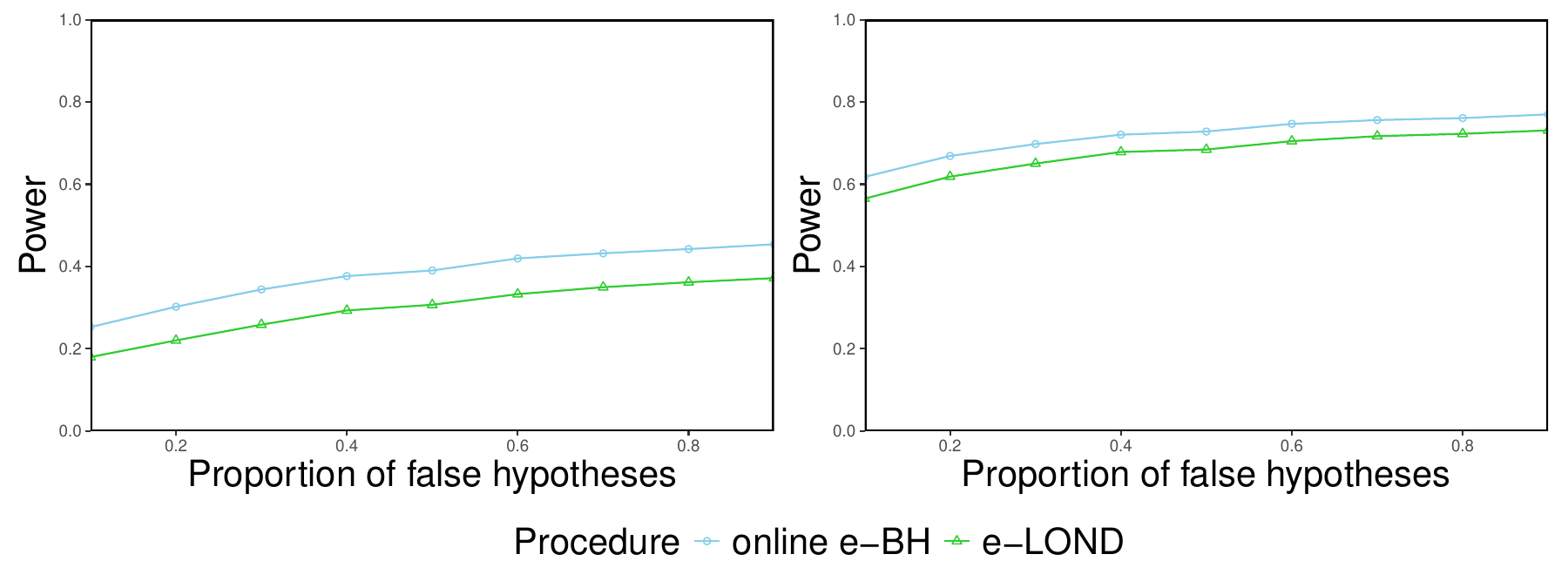}
\caption{Power comparison of online e-BH and e-LOND for different proportions of false hypotheses. In the left (right) plot, the signal of the alternative is weak (strong). 
\label{fig:sim_e_lond} }\end{figure*}
\fi



\begin{remark}
    \citet{xu2023more, xu2024online} improved the e-BH and e-LOND procedure, respectively, using a randomization approach. The idea is to define a randomized e-value as
    $$
    S_{\hat{\alpha}_t}(E_t)\coloneqq \max(\mathbbm{1}\{E_t\geq \hat{\alpha}_t\} E_t, \mathbbm{1}\{U\leq E_t \hat{\alpha}_t\}\hat{\alpha}_t^{-1}),
    $$
    where $\hat{\alpha}_t$ is a potentially data-dependent significance level obtained by e-BH or e-LOND and $U$ a uniform random variable on $[0,1]$ that is independent of $E_t$ and $\hat{\alpha}_t$. \citet{xu2023more} showed that if $E_t$ is an e-value, then $S_{\hat{\alpha}_t}(E_t)$ is a valid e-value as well. Furthermore, $S_{\hat{\alpha}_t}(E_t)\geq \hat{\alpha}_t^{-1}$ iff $E_t\geq \hat{\alpha}_t^{-1} U$, implying that e-BH and e-LOND are uniformly more powerful when applied to $S_{\hat{\alpha}_t}(E_t)$ instead of $E_t$. One could consider applying the same idea to online e-BH. However, note that we do not have one data dependent level with the e-BH but a sequence of levels that is increasing overtime due to future rejections. Hence, one would need to fix a time at which the stochastic rounding approach is applied. For example, one could set
    $$\hat{\alpha}_t\coloneqq \alpha \gamma_t (k_{t-1}^*+1).$$
    In this case stochastic rounding would increase the probability of rejecting $H_t$ at time t. However, if we do not reject, then $S_{\hat{\alpha}_t}(E_t)=0$ and there would not be any chance to obtain a rejection at a later stage. In this case, online e-BH becomes equivalent to the randomized e-LOND procedure. One could also consider applying stochastic rounding at the (potentially data-adaptive) time $\tau$ at which entire testing process was stopped. In this case, stochastic rounding could only lead to additional rejections at the very end and online e-BH becomes a weighted generalization of the randomized e-BH procedure by \citet{xu2023more}. In the next section, we consider a deterministic rounding approach which also leads to improved e-values. However, in contrast to stochastic rounding, it does not allow to round to data-adaptive levels.

\end{remark}

\subsection{Boosting the online e-BH procedure\label{sec:boosting}}


If additional information about the marginal or joint distribution of the e-values is available, then the online e-BH procedure can be improved. We call this process \enquote{boosting} in line with the same term being used for e-BH \citep{wang2022false}. We demonstrate the following three forms of boosting in this paper.

\begin{enumerate}
    \item Boosting when information about the marginal distribution is available.
    \item Boosting when information about the marginal distribution is available and the e-values are locally dependent.
    \item Boosting when information about the marginal distribution is available and the e-values are positive regression dependent on a subset (PRDS).
\end{enumerate}

While the first and third form are online extensions of the boosting techniques by \citet{wang2022false}, the second form of boosting is new and particularly useful in the online setting. Online e-BH controls the SupFDR under the first two boosting types and the OnlineFDR under the third boosting type. In the remaining part of the section, we describe the first form of boosting, which is the cornerstone for all three methods. The details for the second and third form are provided in Appendix~\ref{sec:appn_boosting}.

We first reformulate the boosting method by \citet{wang2022false} which makes it easier to use it for our purposes; however, our approach is equivalent to theirs. Let $K$ be the number of hypotheses. Define the truncation function $T:[0,\infty]\to [0,K/\alpha]$ as
\ifarxiv
\begin{align}
    T(x)\coloneqq \sum_{k=1}^K \mathbbm{1}\left\{ \frac{K}{k\alpha} \leq x < \frac{K}{(k-1)\alpha}\right\} \frac{K}{k\alpha} \quad \text{with } T(\infty)\coloneqq\frac{K}{\alpha}, \label{eq:truncation_offline}
\end{align}
\else
\begin{align}
    T(x)\coloneqq \sum_{k=1}^K \mathbbm{1}\left\{ \frac{K}{k\alpha} \leq x < \frac{K}{(k-1)\alpha}\right\} \frac{K}{k\alpha},\ T(\infty)\coloneqq\frac{K}{\alpha}, \label{eq:truncation_offline}
\end{align}
\fi
where we again use the convention $1/0=\infty$.
Note that $T(x)$ only takes values in the set $\{K/(\alpha k): k\in \{1,\ldots,K\}\}\cup \{0\}.$ Let $X_1,\ldots X_K$ be any nonnegative random variables (not necessarily e-values). The idea of boosting is that e-BH applied on $X_1,\ldots, X_K$ rejects the same hypotheses as if applied on $T(X_1),\ldots, T(X_K)$. Hence, if $T(X_1),\ldots, T(X_K)$ are valid e-values for $H_1,\ldots, H_K$, e-BH applied on $X_1,\ldots, X_K$ controls the SupFDR (even though the latter may not be e-values). We call such random variables $X_t$ \emph{boosted e-values}. A simple approach is to set $X_t=b_tE_t$, where $E_t$ is a valid e-value, for some boosting factor $b_t\geq 1$. Note that such a boosting factor always exists, since $T(E_t)\leq E_t$.

For a given $(\gamma_t)_{t\in \mathbb{N}}$ we can simply extend this approach to the online e-BH procedure by defining the truncation function for the t-th e-value as 
\ifarxiv
\begin{align}
T_t(x)\coloneqq \sum_{k=1}^\infty \mathbbm{1}\left\{ \frac{1}{k\alpha\gamma_t} \leq x < \frac{1}{(k-1)\alpha\gamma_t}\right\} \frac{1}{k\alpha \gamma_t} , \quad \text{with } T_t(\infty)\coloneqq\frac{1}{\alpha\gamma_t} , \label{eq:truncation_online}
\end{align}
\else
\begin{align}
T_t(x)\coloneqq \sum_{k=1}^\infty \mathbbm{1}\left\{ \frac{1}{k\alpha\gamma_t} \leq x < \frac{1}{(k-1)\alpha\gamma_t}\right\} \frac{1}{k\alpha \gamma_t}, \label{eq:truncation_online}
\end{align}
with $T_t(\infty)\coloneqq\frac{1}{\alpha\gamma_t}$
\fi
 and the convention $0 \cdot \infty =0$.
 In case of $\gamma_t=1/K$, $t\leq K$, we have $T_t=T$ for all $t\leq K$. We capture the result about boosted e-values in the following proposition.


\begin{proposition}\label{prop:boosted}
    Let $X_1, X_2, \ldots$ be a sequence of nonnegative random variables such that $\mathbb{E}_{H_t}[T_t(X_t)]\leq 1$ for all $t\in \mathbb{N}$. Then the online e-BH procedure applied to $X_1, X_2, \ldots$ controls the SupFDR at level $\alpha$.
\end{proposition}

Due to the infinite sum in \eqref{eq:truncation_online} it might be difficult to determine boosted e-values exactly. In Appendix~\ref{sec:boosting_arb}, we propose two different approaches to handle this and demonstrate the application in a Gaussian testing setup.

A power comparison of online e-BH with boosted and non-boosted e-values is shown in Figure~\ref{fig:sim_boosted}. The left plot shows the results for weak signals and the right plot for strong signals. Boosting improves the power in both cases.

\ifarxiv
\begin{figure*}[h!]
\centering
\includegraphics[width=0.8\textwidth]{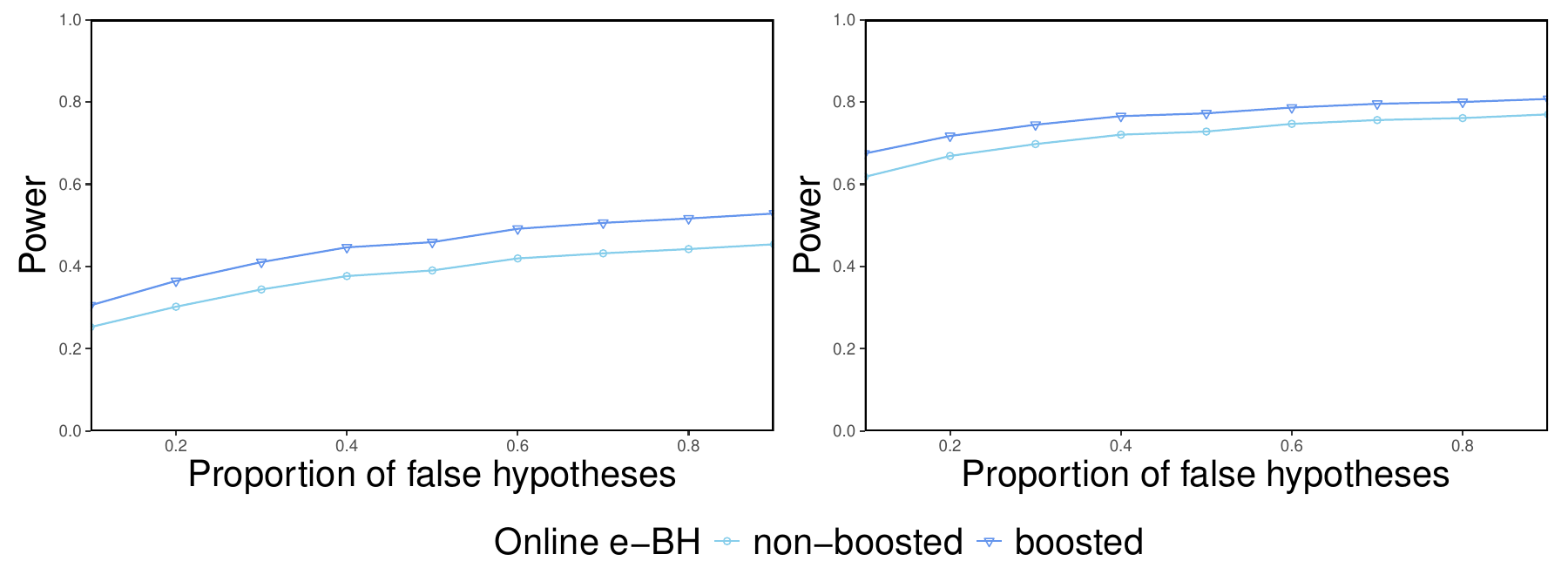}
\caption{Power comparison of online e-BH with non-boosted and boosted e-values for different proportions of false hypotheses. In the left (right) plot, the signal of the alternative is weak (strong). The simulation setup is described in Appendix~\ref{sec:sim_setup}. \label{fig:sim_boosted} }\end{figure*}
\else
\begin{figure*}[h!]
\centering
\includegraphics[width=14cm]{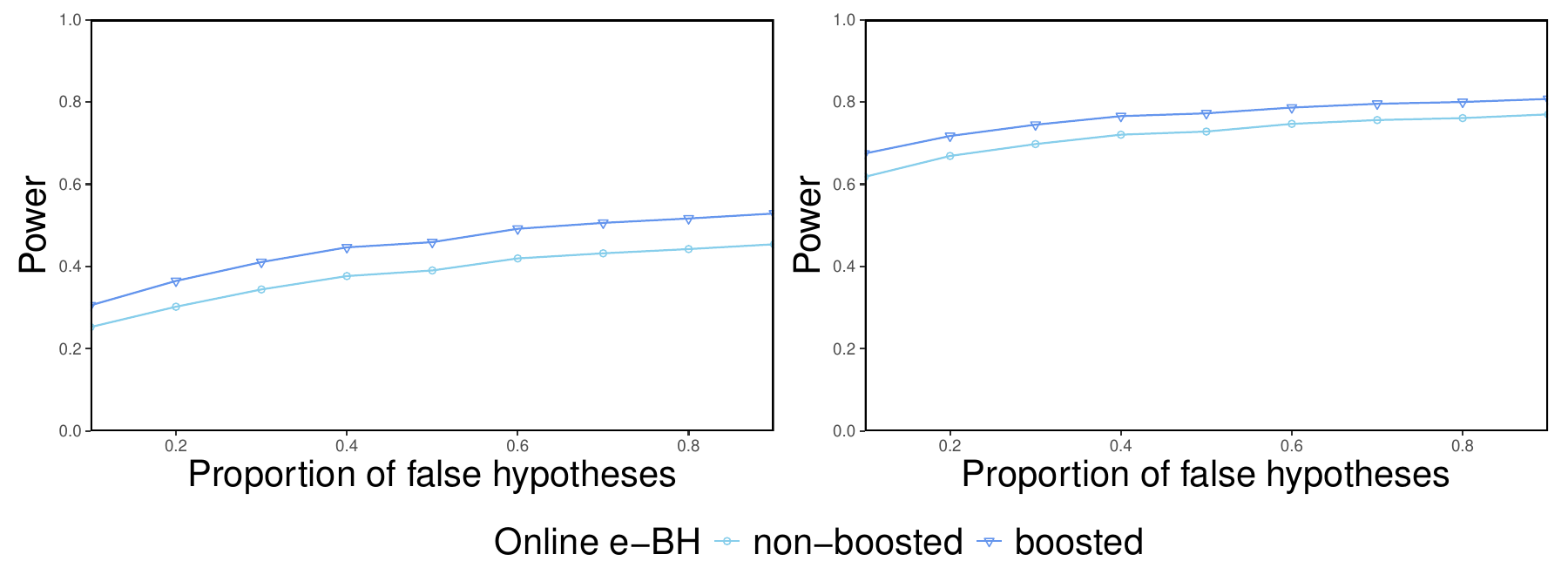}
\caption{Power comparison of online e-BH with non-boosted and boosted e-values for different proportions of false hypotheses. In the left (right) plot, the signal of the alternative is weak (strong). The simulation setup is described in Appendix~\ref{sec:sim_setup}. \label{fig:sim_boosted} }\end{figure*}
\fi

\subsection{Online e-BH with p-values\label{sec:BY}}

In the following, we consider applying the online e-BH procedure to $\psi_1(P_1), \psi_2(P_2), \ldots$, where each $P_t$ is a p-value for $H_t$ and $\psi_t:[0,1]\to [0,\infty]$, $t\in \mathbb{N}$ is a function with $\psi_t(0)=\infty$. \citet{wang2022false} called $\psi_t$ a \textit{decreasing transform} and (for simplicity) assumed that $\psi_1=\psi_2=\ldots=\psi$ for a strictly decreasing and continuous function $\psi$. Here, we allow a different function $\psi_t$ for each hypothesis $H_t$ and only assume that $\psi_t$ is nonincreasing and left-continuous. To clarify this change of definition, we use the term \textit{nonincreasing transform} instead. Define the generalized inverse of a nonincreasing transform $\psi_t$ as
$$
\psi_t^{-1}(x)\coloneqq\max\{u\in [0,1]: \psi_t(u)\geq x\}.
$$
With this, boosting (see Section~\ref{sec:boosting}) can be applied to derive the following result, which is formally proven in the Appendix~\ref{appn:proof_p-vals}.
\begin{proposition}\label{prop:p-values}
    Let $\psi_1, \psi_2,\ldots$ be nonincreasing transforms.
    \begin{enumerate}[label=(\roman*)]
        \item If for every $t$ it holds that
    \begin{align}
        \sum_{k=1}^\infty \frac{1}{k\alpha \gamma_t}\left(\psi_t^{-1}\left(\frac{1}{k \alpha \gamma_t}\right)- \psi_t^{-1}\left(\frac{1}{(k-1) \alpha \gamma_t}\right)\right)\leq 1, \label{eq:condition_p}
    \end{align}
    then the online e-BH procedure applied to the random variables $\psi_1(P_1),\psi_2(P_2), \ldots$ controls the SupFDR at level $\alpha$. \label{bull:prop_p_1}
    \item If the p-values $P_1,P_2,\ldots$ are PRDS (see Appendix~\ref{sec:PRDS} for a formal definition of PRDS) and for all $t$ it holds that
    \begin{align}
        \sup_{k\in \mathbb{N}} \frac{1}{k \alpha \gamma_t} \psi_t^{-1}(1/(k\alpha \gamma_t)) \leq 1,
        \label{eq:condition_p_PRDS}
    \end{align}
    then the online e-BH procedure applied to the random variables $\psi_1(P_1),\psi_2(P_2), \ldots$ controls the OnlineFDR at level $\alpha$.
    \end{enumerate}
\end{proposition}

In the offline case, the Benjamini-Hochberg (BH) procedure \citep{benjamini1995controlling} is equivalent to the e-BH procedure applied to the random variables $1/P_1,\ldots, 1/P_k$ \citep{wang2022false}. One could derive an online BH procedure in the same way by setting the decreasing transforms to $\psi_t(x)=1/x=\psi_t^{-1}(x)$. Proposition~\ref{prop:p-values} immediately implies that this online BH procedure controls the OnlineFDR, if the p-values are PRDS. In the next section, we study the behavior of the online BH procedure more deeply.

\begin{remark}\label{remark:online_BH} \citet{benjamini2001control} introduced an offline p-value based multiple testing procedure with FDR control (BY procedure) and \citet{blanchard2008two} generalized this using a shape function $\beta$ (BR procedure). Let $\beta (k)=\int_{0}^k x \ d\nu (x)$, $k\in \{1,\ldots, K\}$, for some probability measure $\nu$ on $(0,\infty)$.  \citet{blanchard2008two} proved that FDR control is guaranteed under arbitrary dependence of the p-values, if the hypotheses with the $k^*$ smallest p-values are rejected, where
\begin{align}
k^*=\max\left\{k\in \{1,\ldots,K\}: \sum_{j=1}^K \mathbbm{1}\{P_j\leq \beta(k) \alpha /K\}\geq k\right\}. \label{eq:online_BR}
\end{align}
The BY procedure is obtained by setting $\beta(k)=k/\ell_K$, where $\ell_K\coloneqq\sum_{i=1}^K \frac{1}{i}\approx \log(K)$.
\citet{fisher2022online} generalized the method by \citet{blanchard2008two} to the online ARC case in the same manner as we did for the e-BH procedure by using general weights $(\gamma_t)_{t\in \mathbb{N}}$ instead of $1/K$ for each hypothesis (see \eqref{eq:toad} in Appendix for $d_t=\infty$). We call this the online BR procedure in the following.  However, \citet{fisher2022online} only showed FDR control for the online BR procedure at fixed times. In Appendix~\ref{sec:appn_supfdr}, we show how the online BR procedure can be written as a special case of the online e-BH procedure using Proposition~\ref{prop:p-values} \ref{bull:prop_p_1} and therefore provides SupFDR control.
\end{remark}

\begin{remark} \citet{javanmard2015online} introduced a version of their LOND algorithm with OnlineFDR control under arbitrary dependence and later \citet{zrnic2021asynchronous} generalized this method to the reshaped LOND (r-LOND) procedure. The r-LOND procedure rejects hypothesis $H_t$, if $P_t\leq \alpha_t^{\text{r-LOND}}$, where
\ifarxiv
    $$\alpha_t^{\text{r-LOND}}\coloneqq \alpha \gamma_t \beta_t \left(|R_{t-1}^{\text{r-LOND}}|+1\right) \quad \text{and} \quad R_{t-1}^{\text{r-LOND}}\coloneqq \{i\leq t-1: P_i\leq \alpha_i^{\text{r-LOND}}\} $$
\else
\begin{align} &\alpha_t^{\text{r-LOND}}\coloneqq \alpha \gamma_t \beta_t \left(|R_{t-1}^{\text{r-LOND}}|+1\right) \\ \text{and} \quad &R_{t-1}^{\text{r-LOND}}\coloneqq \{i\leq t-1: P_i\leq \alpha_i^{\text{r-LOND}}\} \end{align}
\fi
    and $\beta_t$ is some shape function. In Appendix~\ref{sec:appn_supfdr} we show that r-LOND could be derived as a special case of a boosted e-LOND procedure and therefore also provides SupFDR control under arbitrary dependence of the p-values.
\end{remark}







\section{The online BH procedure\label{sec:onlineBH}}


In the same manner as for e-values, one can define an online ARC version of the p-value based Benjamini-Hochberg (BH) procedure \citep{benjamini1995controlling}, or the \emph{online BH procedure} for short. Recall, as before, that $(\gamma_t)_{t\in \mathbb{N}}$ is a nonnegative sequence that sums to one. Then, define:
 \begin{align}
 &R_t^{\oBH}\coloneqq \left\{i\leq t: P_i\leq {k_t^* \alpha \gamma_i}\right\}, t\in \mathbb{N}, \text{ where}
 \\
&k_t^* =\max\left\{k\in \{1,\ldots,t\}: \sum_{j=1}^t \mathbbm{1}\{P_j\leq k \alpha  \gamma_j \} \geq k\right\}
\end{align}
with the convention $\max(\emptyset)=0$. Again, note that applying the offline BH procedure to the p-values $P_1,\ldots, P_t$ at each time $t$ would not define an online ARC procedure.

The online BH procedure is a special instance of the TOAD framework under positive dependence by \citet{fisher2022online}, which is obtained if no decision deadlines are specified for the TOAD algorithm, meaning  that $d_t=\infty$ for all $t\in \mathbb{N}$. However, the great generality of the TOAD framework leads to a certain notational overhead and complexity that is unnecessary to define and analyze the online BH procedure, leading to this special case being lost. The online BH procedure, we argue, is a simple and important algorithm which deserves particular emphasis.

 As with the online e-BH procedure, if we set $\gamma_1,\dots,\gamma_K$ to equal $1/K$ and the rest of the $\gamma$ sequence to equal zero, one simply obtains the BH procedure at time $K$, so online BH is a \emph{generalization} of the BH procedure.  Although $\text{OnlineFDR}$ control of the online BH procedure under PRDS is implicitly proven by \citet{fisher2022online}, we also provide a short, direct proof in the next subsection. Furthermore, we prove that the online BH procedure additionally controls the SupFDR in case of PRDN (Section~\ref{sec:supFDR_PRDN}) and WNDN (Section~\ref{sec:neg_dep}) p-values, and $\text{SupFDR}^K$ for arbitrary p-values (\Cref{sec:supFDR_oBH_arb}).

\subsection{OnlineFDR control under positive dependence}\label{sec:ofdr-prds}

The most general condition under which FDR control of the BH procedure is usually proven is \emph{positive regression dependence on a subset} (PRDS)  \citep{benjamini2001control} (see Definition~\ref{def:PRDS} in Appendix~\ref{sec:PRDS}). The online BH procedure also controls the OnlineFDR under PRDS.

 \begin{proposition}
     If the null p-values are PRDS, the online BH procedure ensures that $\textnormal{OnlineFDR} \leq \alpha$.
 \end{proposition}

\begin{proof}
For any $t\in \mathbb{N}$, we have
    $$\text{FDR}_t=\sum_{i\leq t, i\in I_0} \mathbb{E}\left[ \frac{\mathbbm{1}\{P_i\leq {k_t^* \alpha \gamma_i}\}}{k_t^*} \right]\leq \sum_{i\leq t, i\in I_0} \alpha \gamma_i \leq \alpha, $$
    where the first inequality follows from Lemma 1 (b) of \citet{ramdas2019unified}.
\end{proof}

In the next subsection, we prove that online BH controls the SupFDR at a slightly inflated level under a weaker assumption than PRDS. 

\begin{remark}\label{remark:LOND}
    In the same manner as online e-BH uniformly improves e-LOND (Section~\ref{sec:e-LOND}), online BH uniformly improves the LOND procedure introduced by \citet{javanmard2015online} which controls the OnlineFDR under PRDS \citep{zrnic2021asynchronous}. The LOND procedure is exactly defined as e-LOND but with p-values instead of e-values, meaning $H_t$ is rejected if $P_t\leq \alpha \gamma_t (|R_{t-1}^{\text{LOND}}|+1)$ where $|R_{t-1}^{\text{LOND}}|$ is the number of previous rejections. Another uniformly more powerful online procedure than LOND is provided by the LORD algorithm \citep{javanmard2018online, ramdas2017online}, which is only known to control the OnlineFDR under independence of the p-values. In Appendix~\ref{sec:LORD} we compare LORD to our online BH procedure.
\end{remark}

\begin{remark}
It is easy to see from the proof above that online BH even controls the OnlineFDR at level $\alpha \pi_0$. In order to avoid this conservatism in the offline case, Storey and colleagues \citep{storey2002direct, storey2004strong} introduced an adaptive BH procedure by estimating $\pi_0$. In Appendix~\ref{appn:SBH}, we introduce the natural online generalization of this Storey-BH procedure and prove its OnlineFDR control under independence of the p-values. 
\end{remark}


\subsection{SupFDR control under positive dependence\label{sec:supFDR_PRDN}}

We will now show that the online BH procedure can also control the SupFDR, albeit at a slightly inflated level. Note that $R_t^\oBH$ is clearly a self-consistent set (just like online e-BH) --- it is the largest such set when $R$ is restricted to be a subset of $\{1,\ldots,t\}$.
Recall that $\mathcal{R}(\alpha)$ is the set of self-consistent rejection sets. The following theorem provides a weighted version of the \texttt{FDR-linking} theorem (i.e., weighted analog of Theorems 1 and 2 in \citet{su_fdr-linking_theorem_2018}).
\begin{theorem}[Weighted \texttt{FDR-linking} theorem]\label{theo:FDR-linking}
    Under any dependence structure among p-values $(P_t)_{t \in [K]}$ for fixed $K \in \naturals$, we can ensure the following is true:
    \begin{align}
        \expect\left[\sup_{R \in \mathcal{R}(\alpha)}\ \FDP(R)\right] \leq \pi_0\alpha + \pi_0\alpha \int\limits_{\pi_0\alpha}^1 \frac{\FDR_0(x)}{x^2}\ \nd x,
    \end{align} where $\FDR_0(x)$ is the FDR of the weighted BH procedure applied only to the p-values corresponding to the null hypotheses, i.e., $I_0$, with weights $(\gamma_t / \pi_0)_{t \in I_0}$.
\end{theorem}

The proof of Theorem~\ref{theo:FDR-linking} is provided in Appendix~\ref{sec:proofs}. We can use the weighted \texttt{FDR-linking} theorem to prove an upper bound of the expected maximum $\text{FDP}$ over all self-consistent discovery sets under \textit{positive regression dependence within nulls} (PRDN). PRDN is a positive dependence assumption that is restricted to the dependency between null p-values and is therefore weaker than PRDS (see Appendix~\ref{sec:PRDS} for a definition of PRDN).


\begin{proposition}\label{prop:p-self-consist}
    For PRDN p-values $(P_t)_{t \in \naturals}$, we have that
    \ifarxiv
    \begin{align}
        \expect\left[\sup_{\rejset \in \Rcal(\alpha)} \FDP(\rejset)\right] \leq
        \pi_0\alpha(1 + \log((\pi_0\alpha)^{-1})) \leq \alpha(1 + \log(\alpha^{-1})).
    \end{align}
    \else
\begin{align}
        \expect\left[\sup_{\rejset \in \Rcal(\alpha)} \FDP(\rejset)\right] &\leq
        \pi_0\alpha(1 + \log((\pi_0\alpha)^{-1})) \\ &\leq \alpha(1 + \log(\alpha^{-1})).
    \end{align}
    \fi
\end{proposition}
\begin{proof}
    Similar to the proof of Proposition~\ref{prop:e-self-consist}, we first prove the above result restricted to a finite set of $K$ hypotheses. Since under positive regression dependence and the global null, the weighted Simes p-value is a valid p-value, we have that $\FDR_0(x) = x$ under PRDN. Then, we get the upper bounds in the proposition statement for the finite $K$ hypotheses
    by carrying out the integration in the upper bound of the weighted \texttt{FDR-linking} theorem and the fact that the function $f(x) = x(1 + \log(1 / x))$ is increasing in $(0, 1]$. Now, we apply the same limit argument as seen in Proposition~\ref{prop:e-self-consist} to show this upper bound extends to the infinite hypotheses case as well. Thus, we have shown our desired result.
\end{proof}

The SupFDR control of the online BH procedure is now a direct result of Proposition~\ref{prop:p-self-consist}.
\begin{theorem}\label{thm:onlinebh-supfdr}
The online BH procedure ensures that $\supFDR \leq \alpha(1 + \log(\alpha^{-1}))$ when the p-values $(P_t)_{t\in \mathbb{N}}$ are PRDN.
\end{theorem}

Note that the bound in \Cref{thm:onlinebh-supfdr} is sharp in the following sense. 
\begin{theorem}\label{thm:sharp-supfdr}
  For every $\varepsilon > 0$, there exists a $K_0 \in \naturals$, and $\alpha' \in (0, 1]$ such that, for online BH, $$\supFDR \geq \StopFDR > (1 - \varepsilon)\alpha(1 + \log(\alpha^{-1}))$$ for all $\alpha \leq \alpha'$ where the first $K_0$ hypotheses are null and have p-values that are i.i.d.\ uniform random variables, and the remaining hypotheses are non-null with p-values set to 0. Note that this p-value distribution satisfies the PRDN conditon, meaning that the SupFDR upper bound in \Cref{thm:onlinebh-supfdr} is unimprovable.
\end{theorem}
We defer the proof of this result to Appendix~\ref{sec:appn_obh_sharp}. 

Notably, the p-value configuration constructed in Theorem~\ref{thm:sharp-supfdr} ensures that all p-values are independent. Hence, it is not even possible to improve the SupFDR bound in \Cref{thm:onlinebh-supfdr} under stronger assumptions like PRDS or independence.





\subsection{SupFDR control under negative dependence\label{sec:neg_dep}}
We can also show that online BH controls the SupFDR under some kind of negative dependence, following a similar approach to \citet{chi_multiple_testing_2024}. Like positive dependence, there are many notions of negative dependence. A comprehensive overview of the different notions and a comparison between these is given by \citet{chi_multiple_testing_2024}. In addition, they showed that the BH procedure controls the FDR at a slightly inflated level, if the p-values are \emph{weakly negatively dependent on the nulls (WNDN)} (see Appendix~\ref{sec:PRDS} for a formal definition of WNDN). The following result can be shown about weighted self-consistent procedures under negative dependence.
\begin{proposition}\label{prop:p-self-consist-neg}
    For WNDN p-values $(P_t)_{t \in \naturals}$, we have that
    \begin{align}
        \expect\left[\sup_{\rejset \in \Rcal(\alpha)} \FDP(\rejset)\right] \leq
        \alpha(3.18 + \log(\alpha^{-1})).
    \end{align}
\end{proposition}
The above result, when restricted to a finite set of hypotheses, follows from the proof of Theorem 19 in \citet{chi_multiple_testing_2024}, as the proof utilizes an application of the \texttt{FDR-linking} theorem with a bound on the type I error of the Simes p-value under WNDN. Prop. 12 in \citet{chi_multiple_testing_2024} shows that the weighted Simes p-value has identical type I error under WNDN as unweighted Simes, so we can apply the weighted \texttt{FDR-linking} theorem to obtain the upper bounds in the proposition for finite hypotheses using an identical argument to the aforementioned proof.
Combining this result with the limit argument seen in Proposition~\ref{prop:p-self-consist}, we can obtain the full result for an infinite stream of hypotheses.

The following theorem shows that online BH even controls SupFDR under WNDN at an inflated level and is a direct result of Proposition~\ref{prop:p-self-consist-neg}.

\begin{theorem}\label{theo:neg}
    The online BH procedure ensures that $\supFDR \leq \alpha(3.18 + \log(\alpha^{-1}))$ when the p-values $(P_t)_{t\in \mathbb{N}}$ are WNDN.
\end{theorem}
The sharpness result in \Cref{thm:sharp-supfdr} holds on an instance of independent p-values (which are WNDN) that achieve SupFDR and StopFDR which approaches $\alpha(1 + \log(\alpha^{-1}))$. Thus, the dependence on $\log(\alpha^{-1})$ cannot be improved, but it may still be possible to show a sharper constant than 3.18.

Propositions~\ref{prop:p-self-consist} and~\ref{prop:p-self-consist-neg} directly imply that the LOND algorithm \citep{javanmard2015online, zrnic2021asynchronous} (see Remark~\ref{remark:LOND}) and the TOAD algorithm for positively dependent p-values \citep{fisher2022online} (see \eqref{eq:toad} in Appendix~\ref{sec:e-TOAD} with $\beta_t(k)=k$) also have the prescribed SupFDR bounds, since both also satisfy the weighted self-consistency condition.

\begin{proposition}\label{prop:supFDR_LOND_TOAD}
    LOND and TOAD ensure that $\supFDR \leq \alpha(1 + \log(\alpha^{-1}))$ when $(P_t)_{t\in \mathbb{N}}$ are PRDN, and $\supFDR \leq \alpha(3.18 + \log(\alpha^{-1}))$ when $(P_t)_{t\in \mathbb{N}}$ are WNDN.
\end{proposition}

\subsection{$\text{SupFDR}^{\boldsymbol{K}}$ control under arbitrary dependence\label{sec:supFDR_oBH_arb}}

In Remark~\ref{remark:online_BH} we mentioned that the online BR procedure controls the SupFDR, which is formally proven in Appendix~\ref{sec:appn_supfdr}. Suppose we have some fixed number $K\in \mathbb{N}$ of hypotheses and let $\beta(k)=k/\ell_K$ for $k\in \{1,\ldots,K\}$, where $\ell_K=\sum_{i=1}^K \frac{1}{i}\approx \log(K)$. It follows immediately, that online BR applied at level $\alpha \ell_K$ rejects exactly the same hypotheses as online BH applied at level $\alpha$. Hence, online BH controls the $\text{SupFDR}^K$ at the inflated level $\alpha \ell_K$.

\begin{theorem}\label{theo:oBH_arb}
    The online BH procedure ensures that $\supFDR^K \leq \alpha \ell_K$ for all $K\in \mathbb{N}$ under arbitrarily dependent p-values.
\end{theorem}

Online BH does not control the SupFDR under arbitrary dependence if the number of hypotheses is infinite. To see this, note that \citet[Theorem 5.1 (iii)]{guo2008control} showed that there exist p-value configurations such that the FDR of the BH procedure equals exactly $\min(\alpha \ell_K ,1)$, where $K$ is the number of total hypotheses. Hence, for every $\alpha$ there is a $K$ such that the FDR of the BH procedure equals $1$. Since the BH procedure is a special case of online BH for every $K$, the online BH procedure can neither control the OnlineFDR, the StopFDR nor the SupFDR at a level below $1$ under arbitrary dependence.

\section{Related literature}

In the following, we give a brief overview of the existing online multiple testing literature divided according to the different settings and assumptions considered in this paper.

\textbf{Independence.} \hspace{0.3cm }The online multiple testing framework was introduced by \citet{foster2008alpha} and initially focused on online FDR control at fixed times and with independent p-values \citep{javanmard2018online, ramdas2018saffron, ramdas2017online}. The most popular procedures in this setting are the LORD algorithm, which was introduced by \citet{javanmard2018online} and later generalized and improved by \citet{ramdas2017online}, and the SAFFRON algorithm, which was introduced by \citet{ramdas2018saffron}. We compare the (improved) LORD algorithm to online BH in Appendix~\ref{sec:LORD} and SAFFRON to online Storey-BH in Appendix~\ref{appn:SBH}. We conclude that the existing methods and our online ARC procedures lead to similar power, however, our online ARC procedures provide a more balanced and better interpretable weighting. In addition, online BH even controls the OnlineFDR under PRDS and allows to control the SupFDR at a slightly inflated level. Such results are not known for LORD nor SAFFRON. 

\textbf{PRDS.} \hspace{0.3cm }
The LOND algorithm was introduced by \citet{javanmard2015online} and proven to control the OnlineFDR under PRDS by \citet{zrnic2021asynchronous}. The TOAD algorithm (for PRDS p-values) by \citet{fisher2022online} generalizes LOND by allowing delayed decisions. We proved that LOND and TOAD also control the SupFDR at a slightly inflated level (under PRDN and WNDN), allowing them to stop data-adaptively at any time (Section~\ref{sec:neg_dep}).

\textbf{Arbitrary dependence.} \hspace{0.3cm }
\citet{javanmard2015online} introduced a modified version of the LOND algorithm that provides OnlineFDR control under arbitrary dependence, which was later generalized by \citet{zrnic2021asynchronous} to the r-LOND algorithm. With TOAD (for arbitrarily dependent p-values), \citet{fisher2022online} introduced a modification of r-LOND that allows to incorporate decision deadlines.
Recently, \citet{xu2024online} introduced e-LOND, providing OnlineFDR control with arbitrarily dependent e-values. In this paper we showed that r-LOND (Appendix~\ref{sec:appn_supfdr}), TOAD (Appendix~\ref{sec:e-TOAD}) and e-LOND (Section~\ref{sec:e-LOND}) even provide SupFDR control without any additional assumptions and therefore can be stopped at arbitrary stopping times. 

\textbf{FDR control at stopping times.} \hspace{0.3cm }
\citet{xu2022dynamic} were the first considering FDR control at arbitrary stopping times. They introduced SupLORD and showed that it provides valid SupFDR control. However, they rely on the assumption that the null p-values are valid conditional on all past p-values, which is a strong assumption close to independence \citep{fisher2024online}. Other works also considered FDR control at particular data-adaptive stopping times \citep{zrnic2021asynchronous, fisher2022online}, all relying on some assumption of independence or conditional validity. Our paper is the first bringing together arbitrary/positive/negative dependence and stopping times.

\textbf{Batching and delayed decisions.} \hspace{0.3cm }
\citet{zrnic2020power} and \citet{fisher2022online} explored the possibility of increasing power of online procedures by batching multiple hypotheses and allowing delayed decisions, respectively. They introduced online frameworks where decisions, both acceptances and rejections, can be reversed until some prespecified decision deadlines. Although \citet{fisher2022online} noted that his TOAD algorithm will never turn a rejected hypothesis into an accepted one, the online ARC framework was not explicitly mentioned before. 



\textbf{E-values.} \hspace{0.3cm }
Multiple testing with e-values was initially considered in the offline setting. \citet{wang2022false} provided with e-BH an analog of the popular Benjamini-Hochberg (BH) procedure \citep{benjamini1995controlling}. \citet{vovk2023confidence, vovk_e-values_calibration_2021} considered merging functions for e-values and applying the Closure Principle with these. Recently, e-values were also brought up in online multiple testing. \citet{xu2024online} introduced the aforementioned e-LOND and \citet{fischer2024online-e} considered online closed testing with e-values, proving that e-values are essential for an online true discovery guarantee.

\section{Conclusion}
Online ARC multiple testing is a flexible and powerful concept. It allows to test hypotheses one at a time, while ensuring that a rejection remains, regardless of the future. However, by changing the decision for accepted hypotheses based on future information online ARC procedures are more powerful than classical online procedures.
We showed that the e-BH and BH procedure possess natural online ARC generalizations.

The online e-BH method controls the SupFDR under arbitrary dependence between the e-values. This allows to apply the online e-BH procedure at any e-values and to stop data-adaptively at any time while ensuring valid FDR control. This offers a much higher flexibility than classical online FDR control at fixed times. We proved that the same flexibility is provided by the existing online procedures r-LOND \citep{zrnic2021asynchronous}, e-LOND \citep{xu2024online}, online BR \citep{blanchard2008two, fisher2022online} and TOAD for arbitrarily dependent p-values \citep{fisher2022online}. Previously, data-adaptive stopping was only proven under some independence assumption about the joint distribution of the test statistics.

The p-value based online BH procedure controls the SupFDR at a slightly inflated level when the p-values are PRDN or WNDN. The same holds for the LOND \citep{javanmard2018online} and the TOAD (for positively dependent p-values) \citep{fisher2022online} algorithm.

In Table~\ref{tab:summary} we list all online procedures with proven SupFDR control under different dependence assumptions between the p-values or e-values.

\ifarxiv
\begin{table*}[h!]
    \centering
    \resizebox{\textwidth}{!}{%
    \begin{tabular}{c c c c c c}
      & & Arb. dep. & Pos. dep. & Neg. dep. & Cond. val. \\ \hline
        \multirow{3}{*}{\rotatebox[origin=c]{90}{e-value}} & Online e-BH & Section~\ref{sec:supFDR_eBH} & $\implies$ & $\implies$ & $\implies$ \\
       & e-LOND \citep{xu2024online} & Section~\ref{sec:e-LOND} & $\implies$ & $\implies$ & $\implies$ \\
      &  e-TOAD & Section~\ref{sec:e-TOAD} & $\implies$ & $\implies$ & $\implies$ \\
         \hline
        \multirow{7}{*}{\rotatebox[origin=c]{90}{p-value}} &  Online BH &  - & $\text{Section~\ref{sec:supFDR_PRDN}}^*$ & $\text{Section~\ref{sec:neg_dep}}^*$ & - \\
       &  Online BR & Section~\ref{sec:appn_supfdr} & $\implies$ &  $\implies$ & $\implies$ \\
       &  r-LOND \citep{zrnic2021asynchronous} & Section~\ref{sec:appn_supfdr} & $\implies$ & $\implies$ & $\implies$\\
       &  TOAD (arb. dep.) \citep{fisher2022online} & Section~\ref{sec:e-TOAD} & $\implies$ & $\implies$ & $\implies$ \\
       &  TOAD (pos. dep.) \citep{fisher2022online} & - & $\text{Section~\ref{sec:neg_dep}}^*$ & $\text{Section~\ref{sec:neg_dep}}^*$ & - \\
       &  LOND \citep{javanmard2018online} & - & $\text{Section~\ref{sec:neg_dep}}^*$ & $\text{Section~\ref{sec:neg_dep}}^*$ & - \\
       &  SupLORD \citep{xu2022dynamic} & - & - & - & \citet{xu2022dynamic} \\
        \hline
    \end{tabular}}
    \caption{Procedures with proven SupFDR control under different dependence assumptions. The \enquote{implies} symbol indicates that the guarantee is implied by control under arbitrary dependence and the asterisk that the SupFDR guarantee is provided at a slightly inflated level.}
    \label{tab:summary}
\end{table*}
\else
\begin{table*}[h!]
    \centering
    \resizebox{14.5cm}{!}{%
    \begin{tabular}{c c c c c c}
      & & Arb. dep. & Pos. dep. & Neg. dep. & Cond. val. \\ \hline
        \multirow{3}{*}{\rotatebox[origin=c]{90}{e-value}} & Online e-BH & Section~\ref{sec:supFDR_eBH} & $\implies$ & $\implies$ & $\implies$ \\
       & e-LOND \citep{xu2024online} & Section~\ref{sec:e-LOND} & $\implies$ & $\implies$ & $\implies$ \\
      &  e-TOAD & Section~\ref{sec:e-TOAD} & $\implies$ & $\implies$ & $\implies$ \\
         \hline
        \multirow{7}{*}{\rotatebox[origin=c]{90}{p-value}} &  Online BH &  - & $\text{Section~\ref{sec:supFDR_PRDN}}^*$ & $\text{Section~\ref{sec:neg_dep}}^*$ & - \\
       &  Online BR & Section~\ref{sec:appn_supfdr} & $\implies$ &  $\implies$ & $\implies$ \\
       &  r-LOND \citep{zrnic2021asynchronous} & Section~\ref{sec:appn_supfdr} & $\implies$ & $\implies$ & $\implies$\\
       &  TOAD (arb. dep.) \citep{fisher2022online} & Section~\ref{sec:e-TOAD} & $\implies$ & $\implies$ & $\implies$ \\
       &  TOAD (pos. dep.) \citep{fisher2022online} & - & $\text{Section~\ref{sec:neg_dep}}^*$ & $\text{Section~\ref{sec:neg_dep}}^*$ & - \\
       &  LOND \citep{javanmard2018online} & - & $\text{Section~\ref{sec:neg_dep}}^*$ & $\text{Section~\ref{sec:neg_dep}}^*$ & - \\
       &  SupLORD \citep{xu2022dynamic} & - & - & - & \citet{xu2022dynamic} \\
        \hline
    \end{tabular}}
    \caption{Procedures with proven SupFDR control under different dependence assumptions. The \enquote{implies} symbol indicates that the guarantee is implied by control under arbitrary dependence and the asterisk that the SupFDR guarantee is provided at a slightly inflated level.}
    \label{tab:summary}
\end{table*}
\fi

\section*{Acknowledgments}
 LF acknowledges funding by the Deutsche Forschungsgemeinschaft (DFG, German Research Foundation) – Project number 281474342/GRK2224/2. AR was funded by NSF grant DMS-2310718.

\section*{Supplementary Material}
 The \texttt{R} code to reproduce all simulations is available at \url{https://github.com/fischer23/online_e-BH}.

\bibliography{main}
\bibliographystyle{plainnat}

\newpage

\begin{appendix}






Here we give a brief overview of the appendix. We begin with a detailed explanation of boosting under local dependence and PRDS e-values in Section~\ref{sec:appn_boosting}.
In Section~\ref{sec:e-TOAD}, we introduce e-TOAD, an e-value version of the TOAD procedure~\citep{fisher2022online}, which can be applied in online settings with decision deadlines. We also show that e-TOAD and TOAD (for arbitrarily dependent p-values) provide SupFDR control under arbitrary dependence. 
In Section~\ref{sec:LORD}, we compare the LORD algorithm \citep{javanmard2018online, ramdas2017online} with the online BH procedure. Afterwards, we introduce an online version of the Storey-BH procedure \citep{storey2002direct, storey2004strong} and compare it to the SAFFRON algorithm \citep{ramdas2018saffron} (Section~\ref{appn:SBH}). 
In Section~\ref{sec:PRDS}, we discuss the considered notions of positive and negative dependence. In Section~\ref{sec:sim_setup}, we describe the simulation setup used in the paper. In Section~\ref{sec:omitted_proofs}, we provide omitted proofs and derivations for results stated in the main paper. 

\section{Details about boosting with online e-BH\label{sec:appn_boosting}}

In this section, we provide several details regarding the boosting techniques briefly described in Section~\ref{sec:boosting}. In Section~\ref{sec:boosting_arb}, we demonstrate how the infinite sum in \eqref{eq:truncation_online} can be handled. In Section~\ref{sec:local_dep}, we show how the boosting approach can be further improved when information about a local dependence structure is available while maintaining SupFDR control. We show in Section~\ref{sec:boosting_PRDS} how OnlineFDR control can be guaranteed while employing an improved boosting technique for PRDS e-values.

\subsection{Handling the infinite sum in the truncation function used for boosting\label{sec:boosting_arb}}


In this section we propose two approaches to handle the infinite sum in \eqref{eq:truncation_online} and demonstrate their application in a Gaussian testing problem. The first approach is a conservative one. Here, for each $t\in \mathbb{N}$ we define some natural number $s\in \mathbb{N}$ and set
\ifarxiv
\begin{align}
\leftindex^+{T}_t^{s}(x)\coloneqq x\mathbbm{1}\left\{x< \frac{1}{s\alpha \gamma_t} \right\} + \sum_{k=1}^s \mathbbm{1}\left\{ \frac{1}{k\alpha\gamma_t} \leq x < \frac{1}{(k-1)\alpha\gamma_t}\right\} \frac{1}{k\alpha \gamma_t}. \label{eq:boosting_online_+}
\end{align}
\else 
\begin{align}
\leftindex^+{T}_t^{s}(x)&\coloneqq x\mathbbm{1}\left\{x< \frac{1}{s\alpha \gamma_t} \right\} \\ &+ \sum_{k=1}^s \mathbbm{1}\left\{ \frac{1}{k\alpha\gamma_t} \leq x < \frac{1}{(k-1)\alpha\gamma_t}\right\} \frac{1}{k\alpha \gamma_t}. \label{eq:boosting_online_+}
\end{align}
\fi
Obviously, we have $T_t(x)\leq \leftindex^+{T}_t^{s}(x)$ for all $s\in \mathbb{N}$ and therefore $\mathbb{E}_{H_t}[\leftindex^+{T}_t^{s}(X_t)]\leq 1$ implies that $X_t$ is a boosted e-value. We still have that $\leftindex^+{T}_t^{s}(x)\leq x$ such that boosting with $\leftindex^+{T}_t^{s}$ can only increase power compared to no boosting.

Another approach is to set
\ifarxiv
\begin{align}
\leftindex^{-}{T}_t^{s}(x)\coloneqq  \sum_{k=1}^s \mathbbm{1}\left\{ \frac{1}{k\alpha\gamma_t} \leq x < \frac{1}{(k-1)\alpha\gamma_t}\right\} \frac{1}{k\alpha \gamma_t}=\leftindex^+{T}_t^{s}(x)- x\mathbbm{1}\left\{x< \frac{1}{s\alpha \gamma_t} \right\}\label{eq:boosting_online_-}
\end{align}
\else 
\begin{align}
\leftindex^{-}{T}_t^{s}(x)&\coloneqq  \sum_{k=1}^s \mathbbm{1}\left\{ \frac{1}{k\alpha\gamma_t} \leq x < \frac{1}{(k-1)\alpha\gamma_t}\right\} \frac{1}{k\alpha \gamma_t} \\ &=\leftindex^+{T}_t^{s}(x)- x\mathbbm{1}\left\{x< \frac{1}{s\alpha \gamma_t} \right\}\label{eq:boosting_online_-}
\end{align}
\fi
for some natural number $s\in \mathbb{N}$. Note that we have $\leftindex^{-}{T}_t^{s}(x)= T_t(x)$ for all $x\geq 1/(s\alpha \gamma_t)$ and $\leftindex^{-}{T}_t^{s}(x)< T_t(x)$ for all $x<1/(s\alpha \gamma_t)$. Therefore, applying online e-BH to $X_t$ with $\mathbb{E}_{H_t}[\leftindex^-{T}_t^{s}(X_t)]\leq 1$ does not necessarily control SupFDR. However, we can apply online e-BH to $\leftindex^-{T}_t^{s}(X_t)$ instead, which provides SupFDR guarantee since $\leftindex^-{T}_t^{s}(X_t)$ is a valid e-value. Furthermore, note that $T_t[\leftindex^-{T}_t^{s}(x)]=\leftindex^-{T}_t^{s}(x)$ for all $x$, implying that the e-value $\leftindex^-{T}_t^{s}(X_t)$ cannot be further improved by boosting if $\mathbb{E}_{H_t}[\leftindex^-{T}_t^{s}(X_t)]=1$. With this approach we do not necessarily improve non-boosted e-values, since $E_t\leq X_t$ does not imply that  $E_t\leq \leftindex^-{T}_t^{s}(X_t)$, where $\leq$ is to be read in an almost sure sense. However, in some situations this is even more powerful than boosting with \eqref{eq:truncation_online}. For example, suppose we know in advance that the total number of rejections won't be larger than $r$. Then boosting with $\leftindex^-{T}_t^{r}$ is uniformly more powerful than boosting with $T_t$, since e-values smaller than $1/(r\alpha \gamma_t)$ cannot lead to a rejection anyway. If such prior information is not available one could also adapt $s$ to $t$, for example, by setting $s=\lceil 1/\gamma_t \rceil$. In this case we could not reject hypotheses with $X_t<\alpha^{-1}$, however, all the mass of  $\leftindex^-{T}_t^{s}(X_t)$ is put to values greater than $\alpha^{-1}$, increasing the probability for a rejection at larger levels.

\citet{wang2022false} have introduced several concrete examples how the boosted e-values $X_1, X_2,\ldots$ could be chosen in specific testing problems. In the following we extend one of them to the online case.

\begin{example}\label{example:boosting}
    Consider an e-value $E_t$ that is obtained from a likelihood ratio between two normal distributions with variance $1$ but a difference $\delta>0$ in means
    $$
    E_t=\exp(\delta Z_t-\delta^2/2),
    $$
    where $Z_t$ follows a standard normal distribution under the null hypothesis $H_t$. Hence, $E_t$ is a log-normally distributed random variable with parameters $(-\delta^2/2,\delta)$ if $H_t$ is true. Suppose we are looking for a boosted e-value of the type $X_t=b_t E_t$ for some boosting factor $b_t\geq 1$. Using \eqref{eq:boosting_online_+}, we can calculate
    \ifarxiv
    \begin{align}
        \mathbb{E}_{H_t}[\leftindex^+{T}_t^{s}(b_tE_t)]&=b_t\left[1-\Phi\left(\frac{\delta}{2}+\log(s\alpha \gamma_tb_t)/{\delta}\right)\right] \\
        &+\sum_{k=1}^{s} \left[\Phi\left(\delta/2-\log([k-1]\alpha \gamma_t b_t)/\delta\right)-\Phi\left(\delta/2-\log(k\alpha \gamma_t b_t)/{\delta}\right)\right]/({k\alpha \gamma_t}),
    \end{align}
    \else 
    \begin{align}
        &\mathbb{E}_{H_t}[\leftindex^+{T}_t^{s}(b_tE_t)]\\ &=b_t\left[1-\Phi\left(\frac{\delta}{2}+\log(s\alpha \gamma_tb_t)/{\delta}\right)\right] \\
        &+\sum_{k=1}^{s} \{\Phi\left(\delta/2-\log([k-1]\alpha \gamma_t b_t)/\delta\right)\\ &-\Phi\left(\delta/2-\log(k\alpha \gamma_t b_t)/{\delta}\right)\}/({k\alpha \gamma_t}),
    \end{align}
    \fi
    where $\Phi$ is the CDF of a standard normal distribution. Thus, we can set this equation to $1$ and solve for $b_t$ numerically.
    For $\alpha=0.05$, $\gamma_t=0.01$, $\delta=3$ and $s=10$, we obtain a boosting factor of $b_t=1.165$. For increasing $s$, the method becomes less conservative and the boosting factor increases. For $s=100$ we obtain $b_t=1.174$, but increasing $s$ further only leads to marginal improvements. Using \eqref{eq:boosting_online_-}, we obtain
    \ifarxiv
    \begin{align}
        \mathbb{E}_{H_t}[\leftindex^-{T}_t^{s}(b_tE_t)]=\sum_{k=1}^{s} \left[\Phi\left(\delta/2-\log([k-1]\alpha \gamma_t b_t)/{\delta}\right)-\Phi\left(\delta/2-\log(k\alpha \gamma_t b_t)/{\delta}\right)\right]/({k\alpha \gamma_t}).
    \end{align}
    \else 
 \begin{align}
        \mathbb{E}_{H_t}[\leftindex^-{T}_t^{s}(b_tE_t)]&=\sum_{k=1}^{s} \{\Phi\left(\delta/2-\log([k-1]\alpha \gamma_t b_t)/{\delta}\right)\\ &-\Phi\left(\delta/2-\log(k\alpha \gamma_t b_t)/{\delta}\right)\}/({k\alpha \gamma_t}).
    \end{align}
    \fi
    In this case, the same parameters as before yield a boosting factor of $b_t=3.071$. This boosting factor decreases if $s$ increases. For $s=100$, we obtain $b_t=1.73$. The boosting factors obtained by these two approaches converge to the same limit for $s\to \infty$. Note that with the latter approach we need to plugin the e-value $\leftindex^-{T}_t^{s}(b_tE_t)$ instead of $b_t E_t$ into the online e-BH algorithm. However, if we expect the number of rejections to be small, we can increase power substantially using $\leftindex^-{T}_t^{s}$.
\end{example}

In the following section, we show how the boosting approach can be further improved if additional information about the joint distribution of the e-values is available.


\subsection{Boosting under local dependence\label{sec:local_dep}}

In many online applications not all e-values are arbitrarily dependent. Indeed, e-values that lie far in the past often have no influence on the current testing process. For this reason, \citet{zrnic2021asynchronous} introduced a \emph{local dependence} structure that allows arbitrary dependence for e-values close together in time while e-values in the distant past should be independent. Precisely, $(E_t)_{t\in \mathbb{N}}$ are called locally dependent with lags $(L_t)_{t\in \mathbb{N}}$, if for all $t\in I_0$ it holds that
$$
E_t \independent E_{t-L_t-1},\ldots, E_1.
$$
 If $L_t=t-1$, the e-values are arbitrarily dependent. If $L_t<t-1$, the e-value $E_t$ is independent of some of the past e-values.

 In the following, we will show that such local dependence information can be exploited to further boost the e-values and increase power. For this, we define a lag-dependent truncation function.
\begin{align}
T_t^{L_t}(x)&=\frac{1}{(k_{t-L_t-1}^*+1)\alpha\gamma_t} \mathbbm{1}\left\{x\geq \frac{1}{(k_{t-L_t-1}^*+1)\alpha\gamma_t} \right\} \nonumber \\ &+ \sum_{k=k_{t-L_t-1}^*+2}^\infty \mathbbm{1}\left\{ \frac{1}{k\alpha\gamma_t} \leq x < \frac{1}{(k-1)\alpha\gamma_t}\right\} \frac{1}{k\alpha \gamma_t}\\
&= \min\left(T_t(x), \frac{1}{(k_{t-L_t-1}^*+1)\alpha\gamma_t} \right), \label{eq:truncation_local}
\end{align}
 where $T_t$ is the truncation function in \eqref{eq:truncation_online}. Note that $T_t^{L_t}(x)\leq T_t(x)$ for all $x$ and $L_t$. Hence, the boosted e-values using $T_t^{L_t}$ can potentially be larger than the ones obtained by $T_t$. The idea of $T_t^{L_t}$ is that if we know that $k_{t-L_t-1}^*$ of the independent past hypotheses were rejected, $H_t$ can be rejected if $E_t\geq [(k_{t-L_t-1}^*+1)\alpha\gamma_t]^{-1}$. Hence, there is no additional gain if $E_t$ is strictly larger than  $ [(k_{t-L_t-1}^*+1)\alpha\gamma_t]^{-1}$. We can then use this information for improved boosting.

\begin{proposition}\label{prop:boosted_seq}
    Let $X_1, X_2, \ldots$ be a sequence of nonnegative random variables such that $\mathbb{E}_{H_t}[T_t^{L_t}(X_t)]\leq 1$ for all $t\in \mathbb{N}$. Then the online e-BH procedure applied with $X_1, X_2, \ldots$ controls the SupFDR at level $\alpha$.
\end{proposition}
Finding such $X_1, X_2, \ldots$ is particularly simple if a sequence of locally dependent e-values $(E_t)_{t\in \mathbb{N}}$ with lags $(L_t)_{t\in \mathbb{N}}$ and known marginal distributions is available. In this case one can simply calculate $\mathbb{E}_{H_t}[T_t^{L_t}(b_tE_t)|k_{t-L_t-1}^*]$, since $E_t$ is independent of $k_{t-L_t-1}^*$.

Analogously to \eqref{eq:boosting_online_+}  and \eqref{eq:boosting_online_-}, we can also define $\leftindex^{+}{T}_t^{L_t,s}$ and $\leftindex^{-}{T}_t^{L_t,s}$, respectively. In the following we continue Example \ref{example:boosting} to quantify the improvement gain due to local dependence.

\begin{example}\label{example:boosting_local}
Consider Example \ref{example:boosting} but suppose the e-value $E_t$ is independent of the e-values $E_1,\ldots, E_{t-L_t-1}$ for some $L_t\in \{1,\ldots, t-1\}$.  We obtain
\ifarxiv
\begin{align}
        &\mathbb{E}_{H_t}[\leftindex^+{T}_t^{L_t,s}(b_tE_t)|k_{t-L_t-1}^*]\\ &=b_t\left[1-\Phi\left(\delta/{2}+\log(s\alpha \gamma_tb_t)/{\delta}\right)\right] \\
        &+ \left[1-\Phi\left(\delta/2-\log((k_{t-L_t-1}^*+1)\alpha \gamma_t b_t)/\delta\right)\right]/{[(k_{t-L_t-1}^*+1)\alpha \gamma_t]} \\
        &+\sum_{k=k_{t-L_t-1}^*+2}^{s} \left[\Phi\left(\delta/2-\log([k-1]\alpha \gamma_t b_t)/\delta\right)-\Phi\left(\delta/2-\log(k\alpha \gamma_t b_t)/\delta\right)\right]/({k\alpha \gamma_t}).
    \end{align}
\else 
\begin{align}
        &\mathbb{E}_{H_t}[\leftindex^+{T}_t^{L_t,s}(b_tE_t)|k_{t-L_t-1}^*]\\ &=b_t\left[1-\Phi\left(\delta/{2}+\log(s\alpha \gamma_tb_t)/{\delta}\right)\right] \\
        &+ \frac{1-\Phi\left(\delta/2-\log([k_{t-L_t-1}^*+1]\alpha \gamma_t b_t)/\delta\right)}{(k_{t-L_t-1}^*+1)\alpha \gamma_t} \\
        &+\sum_{k=k_{t-L_t-1}^*+2}^{s} \{\Phi\left(\delta/2-\log([k-1]\alpha \gamma_t b_t)/\delta\right) \\ &-\Phi\left(\delta/2-\log(k\alpha \gamma_t b_t)/\delta\right)\}/({k\alpha \gamma_t}).
    \end{align}
\fi
For the parameters $\alpha=0.05$, $\gamma_t=0.01$, $\delta=3$, $s=100$ and $k_{t-L_t-1}^*=2$, this leads to the boosting factor $b_t=1.265$. For a larger number of previous (independent) rejections $k_{t-L_t-1}^*=10$, the boosting factor also increases to $b_t=1.541$. Furthermore, we have
\ifarxiv
\begin{align}
        \mathbb{E}_{H_t}[\leftindex^-{T}_t^{L_t,s}(b_tE_t)|k_{t-L_t-1}^*]&= \mathbb{E}_{H_t}[\leftindex^+{T}_t^{L_t,s}(b_tE_t)|k_{t-L_t-1}^*] - b_t\left[1-\Phi\left(\delta/{2}+\log(s\alpha \gamma_tb_t)/{\delta}\right)\right].
    \end{align}
    \else 
\begin{align}
        &\mathbb{E}_{H_t}[\leftindex^-{T}_t^{L_t,s}(b_tE_t)|k_{t-L_t-1}^*] \\ &= \mathbb{E}_{H_t}[\leftindex^+{T}_t^{L_t,s}(b_tE_t)|k_{t-L_t-1}^*] \\ &- b_t\left[1-\Phi\left(\delta/{2}+\log(s\alpha \gamma_tb_t)/{\delta}\right)\right].
    \end{align}
    \fi
    In this case we obtain a boosting factor of $b_t=1.940$ for $k_{t-L_t-1}^*=2$ and $b_t=2.639$ for $k_{t-L_t-1}^*=10$.
\end{example}

The gain in power obtained by exploiting information about the local dependence structure to improve boosting is illustrated in Figure  \ref{fig:sim_boosted}. It is seen that using information about locally dependent e-values improves power significantly.

\ifarxiv
\begin{figure*}[h!]
\centering
\includegraphics[width=0.8\textwidth]{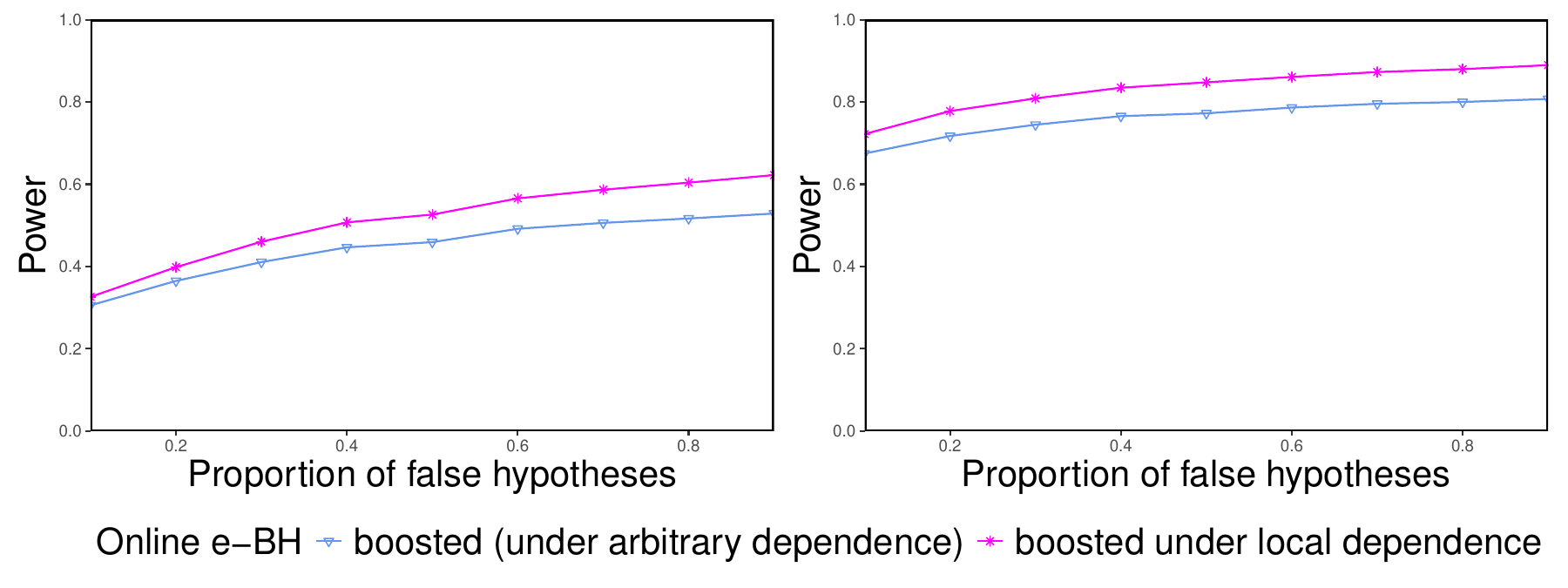}
\caption{Power comparison of online e-BH with boosted and boosted e-values under local dependence for different proportions of false hypotheses. In the left (right) plot the signal of the alternative is weak (strong). The simulation setup is described in Section \ref{sec:sim_setup}. \label{fig:sim_boosted_local} }\end{figure*}
\else 
\begin{figure*}[h!]
\centering
\includegraphics[width=14cm]{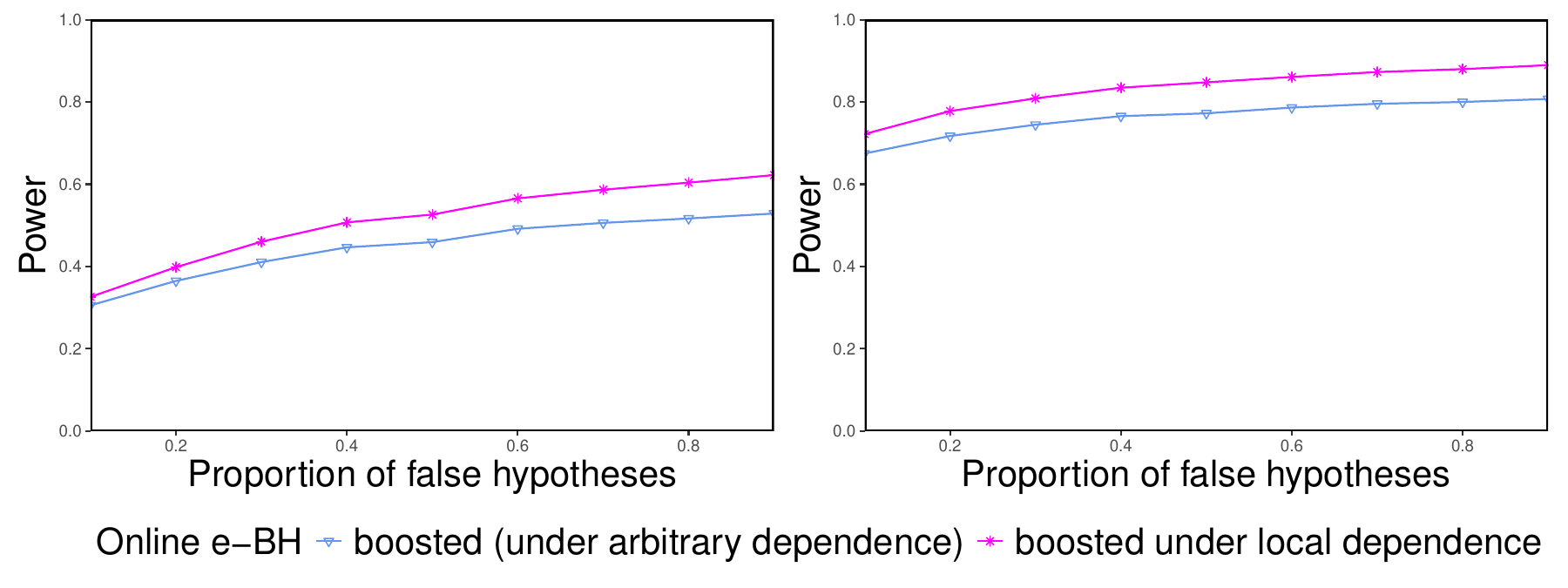}
\caption{Power comparison of online e-BH with boosted and boosted e-values under local dependence for different proportions of false hypotheses. In the left (right) plot the signal of the alternative is weak (strong). The simulation setup is described in Section \ref{sec:sim_setup}. \label{fig:sim_boosted_local} }\end{figure*}
\fi

In the p-value literature, it is usually started with a procedure under independence which is made conservative when local dependence is present. Interestingly, we take the other direction and start with a procedure that works under arbitrary dependence and improve it when information about a local dependence structure is available. A similar behavior for the BH and e-BH procedure was noted by \citet{wang2022false}.

\begin{remark}
    Another advantage of local dependence is that $\gamma_t$ can be measurable with respect to $\mathcal{F}_{t-L_t-1}=\sigma(E_1,\ldots,E_{t-L_t-1})$ and thus depend on the previous data. In fact, to choose $\gamma_t$ predictable it is even sufficient that $E_t$ is valid conditional on $\mathcal{F}_{t-L_t-1}$, which is slightly weaker than independence.
    To see this, just note that for all $t\in I_0$ it holds that \ifarxiv \begin{align}\mathbb{E}[\gamma_t E_t]=\EE[\mathbb{E}[\gamma_t E_t|\mathcal{F}_{t-L_t-1}]] =\EE[\gamma_t\mathbb{E}[ E_t|\mathcal{F}_{t-L_t-1}]]\leq \EE[\gamma_t].\end{align}
    \else 
\begin{align}\mathbb{E}[\gamma_t E_t]&=\EE[\mathbb{E}[\gamma_t E_t|\mathcal{F}_{t-L_t-1}]]\\ &=\EE[\gamma_t\mathbb{E}[ E_t|\mathcal{F}_{t-L_t-1}]]\leq \EE[\gamma_t].\end{align}
    \fi
\end{remark}

\subsection{Boosting with PRDS e-values\label{sec:boosting_PRDS}}

\citet{wang2022false} also introduced a more powerful boosting technique than the one we described in Section \ref{sec:boosting_arb} which works under positive regression dependence on a subset (PRDS) \citep{benjamini2001control} (see Section \ref{sec:PRDS} for a formal definition of PRDS for e-values). In the following we extend this method to the online setting. However, the proof technique of FDR control is different in this case such that SupFDR control can no longer be guaranteed and we only show OnlineFDR control.

Let the truncation function $T_t$ be as in \eqref{eq:truncation_online}. Instead of bounding the expected value of the truncated e-value $T_t(X_t)$, \citet{wang2022false} proposed to bound the probability
\ifarxiv
\begin{align}
\sup_{k\in \mathbb{N}} \frac{1}{k \alpha \gamma_t}\mathbb{P}_{H_t}(X_t\geq 1/(k \alpha \gamma_t)) =\sup_{x\geq 0} x\mathbb{P}_{H_t}(T_t(X_t)\geq x)\leq 1, \label{eq:boost_PRDS}
\end{align}
\else 
\begin{align}
&\sup_{k\in \mathbb{N}} \frac{1}{k \alpha \gamma_t}\mathbb{P}_{H_t}(X_t\geq 1/(k \alpha \gamma_t)) \\&=\sup_{x\geq 0} x\mathbb{P}_{H_t}(T_t(X_t)\geq x)\leq 1, \label{eq:boost_PRDS}
\end{align}
\fi
if the $X_t$, $t\in \mathbb{N}$, are PRDS. In the following, we show that this condition can also be used for boosting in the online case.

\begin{proposition}\label{prop:boost_PRDS}
    Let $X_1,X_2,\ldots$ be a sequence of nonnegative random variables such that \eqref{eq:boost_PRDS} holds for all $t\in \mathbb{N}$. Then the online e-BH procedure applied with $X_1,X_2,\ldots$ controls the OnlineFDR at level $\alpha$, if the random variables $X_1,X_2,\ldots$ are PRDS.
\end{proposition}
\begin{proof}
    Let $t\in \mathbb{N}$ be fixed and $\boldsymbol{X}=(X_1\alpha \gamma_1,\ldots,X_t\alpha \gamma_t)$. Then $f(\boldsymbol{X})=1/(k_t^*\lor 1)$, where $$k_t^*=\max\left\{k\in \{1,\ldots,t\}: \sum_{j=1}^t \mathbbm{1}\{X_j\geq 1/(k \alpha  \gamma_j) \} \geq k\right\},$$ is a nonincreasing function. Furthermore, the range of $f$ is given by $I_f=\{1,1/2, 1/3,\ldots\}$ and
    $$
    \sup\{x\in I_f: x \leq X_i\alpha \gamma_i\}=T_i(X_i)\alpha \gamma_i.
    $$
    With this,  Lemma 1 (ii) in \citep{wang2022false} implies that
    $$
    \mathbb{E}[f(\boldsymbol{X})\mathbbm{1}\{X_i \alpha \gamma_i \geq f(\boldsymbol{X})\}]\leq \sup_{x\geq 0} x \mathbb{P}(T_i(X_i)\alpha \gamma_i\geq x)
    $$
    for all $i\in I_0$, $i\leq t$. Since
    $$
    \sup_{x\geq 0} x \mathbb{P}(T_i(X_i)\alpha \gamma_i\geq x)\leq \alpha \gamma_i \Leftrightarrow \sup_{x\geq 0} x \mathbb{P}(T_i(X_i)\geq x) \leq 1,
    $$
    we obtain
    \begin{align}
    \text{FDR}_t&=\sum_{i\in I_0, i\leq t} \mathbb{E}\left[ \frac{\mathbbm{1}\{ X_i\geq 1/(\alpha \gamma_i (k_t^*\lor 1))}{(k_t^*\lor 1)} \right]
    \\ &\leq \sum_{i\in I_0, i\leq t} \sup_{x\geq 0} x \mathbb{P}(T_i(X_i)\alpha \gamma_i\geq x) \leq \alpha \sum_{i\in I_0, i\leq t} \gamma_i \leq \alpha.
    \end{align}

\end{proof}

By Markov's inequality, $\mathbb{E}_{H_t}[T_t(X_t)]\leq 1$ implies that \eqref{eq:boost_PRDS} is satisfied. Hence, boosting with \eqref{eq:boost_PRDS} is more powerful than boosting with  $\mathbb{E}_{H_t}[T_t(X_t)]\leq 1$.

\section{The e-TOAD procedure\label{sec:e-TOAD}}

First, we review the TOAD procedure by \citet{fisher2022online} under arbitrary dependence. Afterwards, we introduce e-TOAD, an e-value analog of the TOAD procedure and show that e-TOAD provides SupFDR control. By showing that TOAD is a special case of e-TOAD we also conclude SupFDR control for the TOAD procedure.

In the online setup with \enquote{decision deadlines} by \citet{fisher2022online} each $H_t$ comes with a decision deadline $d_t\geq t$, which means that the decision for hypothesis $H_t$ can be changed until time $d_t$. For example, if $d_t=t$ for all $t\in \mathbb{N}$, the classical online setup is obtained.  Furthermore, let $C_t=\{i\leq t: d_i\geq t\}$ be the set of currently \enquote{active} hypotheses.  We will see that TOAD and e-TOAD never turn a rejected hypothesis into an accepted hypothesis and therefore are online ARC procedures.

Initialize $R_0=\emptyset$. The TOAD algorithm under arbitrary dependence \citep{fisher2022online} is recursively defined as follows.
\ifarxiv
\begin{align}
    & R_t^{\text{TOAD}}\coloneqq\left\{i\in \{1,\ldots,t\}:P_i\leq \alpha \gamma_i \beta_i(k^*_{\min(d_i,t)})\right\}, t\in \mathbb{N}, \text{ where } \label{eq:toad} \\
     & k_t^*=|R_{t-1}^{\text{TOAD}}\setminus C_t| + \max\left\{k\leq |C_t|: \sum_{j\leq t, j\in C_t} \mathbbm{1}\{P_j\leq \alpha \gamma_j \beta_j(|R_{t-1}^{\text{TOAD}}\setminus C_t|+k)\} \geq k\right\}. \nonumber
\end{align}
\else 
\begin{align}
    & R_t^{\text{TOAD}}\coloneqq\left\{i\in \{1,\ldots,t\}:P_i\leq \alpha \gamma_i \beta_i(k^*_{\min(d_i,t)})\right\}, \label{eq:toad} \\
    \text{where } & k_t^*=|R_{t-1}^{\text{TOAD}}\setminus C_t| + \max\left\{k\leq |C_t|: \Sigma_{k,t}^{\text{TOAD}} \geq k\right\} \\
    \text{and } &\Sigma_{k,t}^{\text{TOAD}}=\sum_{j\leq t, j\in C_t} \mathbbm{1}\{P_j\leq \alpha \gamma_j \beta_j(|R_{t-1}^{\text{TOAD}}\setminus C_t|+k)\}. \nonumber
\end{align}
\fi
The shape functions $\beta_i$ are defined the same as in Section \ref{sec:BY}. The TOAD procedure under arbitrary dependence is an online generalization of the BR procedure \citep{blanchard2008two}. Note that \citet{fisher2022online} only allowed to use the same shape function $\beta=\beta_i$ for all $i\in \mathbb{N}$. This can cost a lot of power, particularly if $d_t<\infty$. For example, suppose $d_1=1$, which means that the decision for the first hypothesis must be made immediately. Then we could choose $\beta_1(1)=1$, meaning to put all mass on $k_1^*=1$ since this is the only relevant value. This would not be possible if the same shape function must be used for all hypotheses, which usually implies $\beta_1(1)<1$. We will show that TOAD even controls the SupFDR for individual shape functions by proving that it is a special case of e-TOAD.

The e-TOAD algorithm is defined by
\ifarxiv
\begin{align}
    & R_t^{\text{e-TOAD}}\coloneqq\left\{i\in \{1,\ldots,t\}:E_i\geq (\alpha \gamma_i k^*_{\min(d_i,t)})^{-1}\right\}, t\in \mathbb{N}, \text{ where } \\
     & k_t^*=|R_{t-1}^{\text{e-TOAD}}\setminus C_t| + \max\left\{k\leq |C_t|: \sum_{j\leq t, j\in C_t} \mathbbm{1}\{E_j\geq [\alpha \gamma_j (|R_{t-1}^{\text{e-TOAD}}\setminus C_t|+k)]^{-1}\} \geq k\right\}. \nonumber
\end{align}
\else 
\begin{align}
    & R_t^{\text{e-TOAD}}\coloneqq\left\{i\in \{1,\ldots,t\}:E_i\geq (\alpha \gamma_i k^*_{\min(d_i,t)})^{-1}\right\}, \\
    &\text{where }  k_t^*=|R_{t-1}^{\text{e-TOAD}}\setminus C_t| + \max\left\{k\leq |C_t|: \Sigma_{k,t}^{\text{e-TOAD}} \geq k\right\}\\
    &\text{and } \Sigma_{k,t}^{\text{e-TOAD}}=\sum_{j\leq t, j\in C_t} \mathbbm{1}\{E_j\geq [\alpha \gamma_j (|R_{t-1}^{\text{e-TOAD}}\setminus C_t|+k)]^{-1}\}. \nonumber
\end{align}
\fi
Note that e-TOAD becomes online e-BH in case of $d_t=\infty$ for all $t$ and e-LOND if $d_t=t$ for all $t$.

Analogously to Section \ref{sec:online_eBH}, the SupFDR control of the e-TOAD algorithm follows because the e-value $E_i$, $i\in \{1,\ldots,t\}$, of every rejected hypothesis at time $t$ satisfies
$$
E_i\geq 1/(\alpha \gamma_i |R_t^{\text{e-TOAD}}|).
$$
\begin{proposition}
    The e-TOAD procedure provides SupFDR control at level $\alpha$ under arbitrary dependence between the e-values.
\end{proposition}

Define the truncation function as
\ifarxiv
\begin{align}
T_t^{\text{TOAD}}(x)=\sum_{k=1}^{d_t} \mathbbm{1}\left\{ \frac{1}{k\alpha\gamma_t} \leq x < \frac{1}{(k-1)\alpha\gamma_t}\right\} \frac{1}{k\alpha \gamma_t} \quad \text{with } T_t(\infty)=\frac{1}{\alpha\gamma_t}. \label{eq:truncation_toad}
\end{align}
\else 
\begin{align}
T_t^{\text{TOAD}}(x)=\sum_{k=1}^{d_t} \mathbbm{1}\left\{ \frac{1}{k\alpha\gamma_t} \leq x < \frac{1}{(k-1)\alpha\gamma_t}\right\} \frac{1}{k\alpha \gamma_t} \label{eq:truncation_toad}
\end{align}
with $T_t(\infty)=\frac{1}{\alpha\gamma_t}$. 
\fi
\sloppy Then e-TOAD applied to $X_1, X_2,\ldots$ rejects the same hypotheses as if applied to $T_1^{\text{TOAD}}(X_1),T_2^{\text{TOAD}}(X_2),\ldots$~. Therefore, applying e-TOAD to $X_1, X_2,\ldots$ ensures SupFDR control, if $\EE[T_t^{\text{TOAD}}(X_t)]\leq 1$ for all $t\in I_0$. With this, we can prove the following proposition in the exact same manner as we proved Proposition \ref{prop:p-values}.

\begin{proposition}\label{prop:transform_TOAD}
    Let $\psi_1, \psi_2,\ldots$ be nonincreasing transforms  such that for every $t$ it holds that
    \begin{align}
        \sum_{k=1}^{d_t} \frac{1}{k\alpha \gamma_t}\left(\psi_t^{-1}\left(\frac{1}{k \alpha \gamma_t}\right)- \psi_t^{-1}\left(\frac{1}{(k-1) \alpha \gamma_t}\right)\right)\leq 1. \label{eq:condition_p_toad}
    \end{align}
    Then the e-TOAD procedure applied to the random variables $\psi_1(P_1),\psi_2(P_2), \ldots$ controls the SupFDR at level $\alpha$.
\end{proposition}

Define $\psi_t$ such that $\psi_t^{-1}(1/(\alpha \gamma_t k))=\alpha \beta_t(k) \gamma_t$. Then  $\psi_t(P_t)\geq 1/(\alpha \gamma_t k)$ iff $P_t\leq \psi_t^{-1}(1/(\alpha \gamma_t k))=\alpha \beta_t(k) \gamma_t$. Hence, e-TOAD applied to $\psi_1(P_1),\psi_2(P_2),\ldots $ rejects the same hypotheses as TOAD applied to $P_1, P_2,\ldots$ . In the exact same manner as in Proposition \ref{prop:BR_procedure}, we can show that $\psi_1,\psi_2,\ldots$ satisfy Proposition \ref{prop:transform_TOAD}. Hence, TOAD controls the SupFDR.

\begin{proposition}\label{prop:SupFDR_TOAD}
    The TOAD procedure \eqref{eq:toad} provides SupFDR control at level $\alpha$ under arbitrary dependence between the p-values.
\end{proposition}


\section{Online BH versus LORD\label{sec:LORD}}

In this section, we compare the LORD procedure \citep{javanmard2018online, ramdas2017online} with the online BH procedure. 
\citet{javanmard2018online} introduced the LORD algorithm with OnlineFDR control which was later generalized and improved to the LORD++ algorithm by \citet{ramdas2017online}. Here we follow the latter presentation, but call it LORD instead of LORD++ for brevity.

Let $\alpha_1^{\LORD}, \alpha_2^{\LORD}, \ldots$ be individual levels such that 
\ifarxiv
\begin{align}
    \frac{\sum_{i=1}^t \alpha_i^{\LORD} }{|R_t^{\LORD}|\lor 1}\leq \alpha \ \text{for all } t\in \mathbb{N}, \text{ where } R_t^{\LORD}\coloneqq\{i\leq t: P_i\leq \alpha_i^{\LORD}\}.\label{eq:cond_lord}
\end{align}
\else
\begin{align}
    &\frac{\sum_{i=1}^t \alpha_i^{\LORD} }{|R_t^{\LORD}|\lor 1}\leq \alpha \ \text{for all } t\in \mathbb{N}, \label{eq:cond_lord} \\ \text{where } & R_t^{\LORD}\coloneqq\{i\leq t: P_i\leq \alpha_i^{\LORD}\}.
\end{align}
\fi

\citet{ramdas2017online} showed that any online procedure defined by the rejection sets $R_t^{\LORD}$, $t\in \mathbb{N}$, controls the OnlineFDR if the null p-values are independent from each other and the non-nulls; and the thresholds $\alpha_i^{\text{LORD}}$, $i\in \mathbb{N}$, are nonrandom and nondecreasing functions of the past rejection indicators $\mathbbm{1}\{P_1\leq \alpha_1^{\LORD}\}, \ldots, \mathbbm{1}\{P_{i-1}\leq \alpha_{i-1}^{\LORD}\}$.

Let $\alpha_{i,t}^{\oBH}:=k_t^*\alpha \gamma_i$, $t\in \mathbb{N}$ and $i\leq t$, be the individual significance level for hypothesis $H_i$ at time $t$ obtained by the online BH procedure (Section~\ref{sec:onlineBH}). Since $k_t^*=|R_t^{\oBH}|$, it is easy to see that the individual levels $\alpha_{i,t}^{\oBH}$, $i\leq t$, also satisfy the LORD condition \eqref{eq:cond_lord}. However, $\alpha_{i,t}^{\oBH}$ is not a function of only the first $i-1$ rejections, but depends on all rejections up to step $t$. Therefore, online BH is not a special instance of the LORD algorithm, but online BH and LORD are both allowed to spend the same amount of significance level up to each time $t$. For this reason, online BH and LORD should have similar power, and the difference between these procedures lies in a different weighting of the hypotheses (with online BH all individual levels gain from rejections, while with LORD only future levels gain from rejections).

This is verified via simulations in Figure~\ref{fig:sim_LORD}, where we compare online BH with a special instance of the LORD algorithm implemented in the \texttt{R} package \texttt{onlineFDR} \citep{robertson2019onlinefdr}. The simulation setup is described in Section~\ref{sec:sim_setup}. Recall that $(\gamma_t)_{t\in \mathbb{N}}$ is a sequence that sums to one and which defines the weighting of the different hypotheses. In Figure~\ref{fig:sim_LORD}, we set $\gamma_t=q^{t-1}(1-q)$ with $q=0.99$ in the left plot and with $q=0.999$ in the right plot for both procedures. That means, in the right plot the $\gamma$ sequence is decreasing slower. It is easy to see that LORD performs better than online BH in the left plot and worse in the right plot. In addition, online BH with $q=0.999$ performs nearly the exact same as LORD with $q=0.99$. This supports our theoretical reasoning that the difference between these two procedures lies in a different weighting of the hypotheses, but both procedures spend the same amount of significance level. 

We think one advantage of online BH compared to LORD is that its weighting is easier to interpret. If $\gamma_i=2\gamma_j$ for $i\neq j$, then $H_i$ will be tested at twice the level compared to $H_j$ with online BH, regardless of the number of rejections. In contrast, the weighting of LORD strongly depends on the number of rejections in a complex nonlinear manner. Another important advantage of online BH compared to LORD is that online BH controls the OnlineFDR under PRDS, while LORD is only known to control it under independence. In addition, we proved that online BH even controls the SupFDR at slightly inflated levels under PRDN and WNDN. Such results are not known for the LORD procedure. However, of course LORD is an online procedure in the strict sense, which can be desirable in certain applications, while online BH is an online ARC procedure. We conclude that there is much to be gained by relaxing the online requirement to an online ARC requirement: we get a procedure with seemingly comparable power, but with a much stronger suite of type-I error control guarantees. 

\ifarxiv
\begin{figure*}[h!]
\centering
\includegraphics[width=14cm]{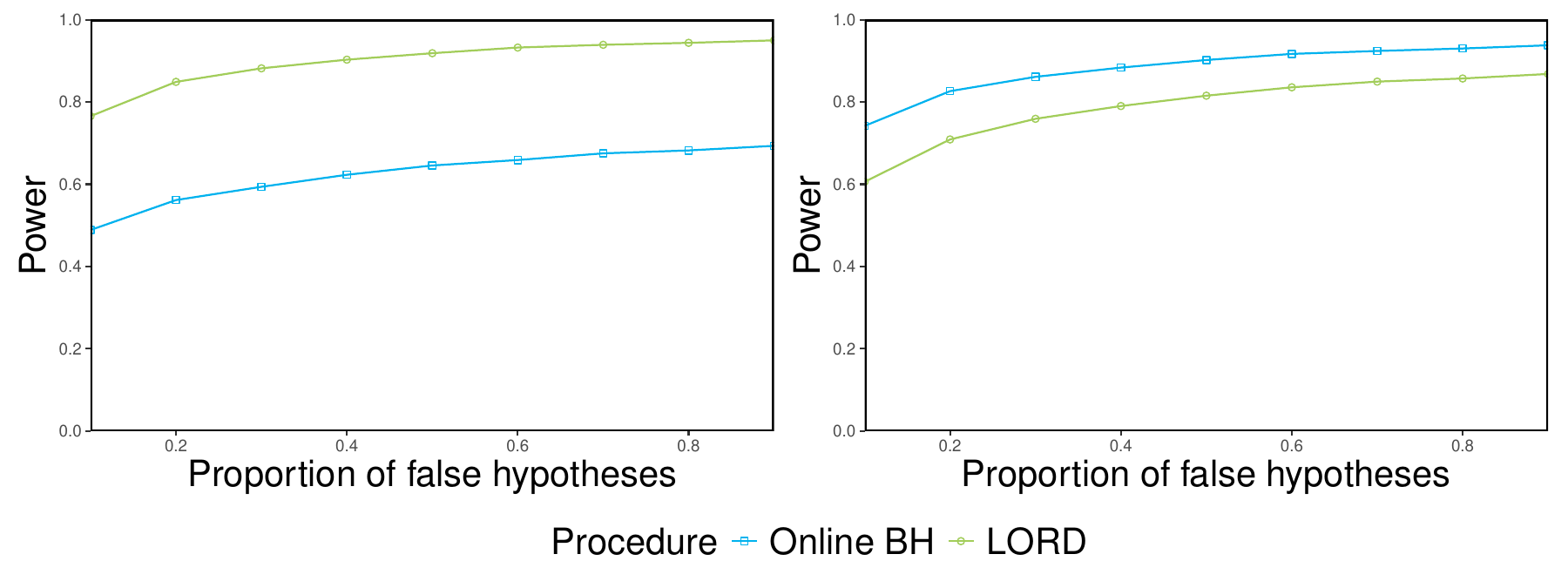}
\caption{Power comparison of online BH and LORD for different proportions of false hypotheses. In the left plot the sequence $(\gamma_t)_{t\in \mathbb{N}}$ decreases fast ($q=0.99$) and in the left plot it decreases slow ($q=0.999$). The simulation setup is described in Section \ref{sec:LORD}. \label{fig:sim_LORD} }\end{figure*}
\fi

\section{The online Storey-BH procedure\label{appn:SBH}}

Similarly as for e-BH and BH, one can define an online version of the Storey-BH procedure \citep{storey2002direct, storey2004strong}, which we call just SBH in the following. In this section, we introduce the online SBH procedure, prove its OnlineFDR guarantee and compare it with the existing SAFFRON \citep{ramdas2018saffron} procedure.

Recall that $(\gamma_t)_{t\in \mathbb{N}}$ is a sequence that sums to one. In addition, let $\lambda\in [\alpha, 1)$ be a user-defined constant and \begin{align}\hat{\pi}_0^t=\frac{\max_{i\in \mathbb{N}} \gamma_i +\sum_{i=1}^{\infty} \gamma_i \mathbbm{1}\{P_i>\lambda \lor i>t\}}{(1-\lambda)}. \label{eq:estimate_prop}\end{align} Then the online SBH method is defined by the rejection sets
\ifarxiv
 \begin{align}
 R_t^{\oSBH}&\coloneqq \left\{i\leq t: P_i\leq {\min(k_t^* \alpha \gamma_i/\hat{\pi}_0^t, \lambda)}\right\}, t\in \mathbb{N},
 \\
 \text{where  }
k_t^* &=\max\left\{k\in \{1,\ldots,t\}: \sum_{j=1}^t \mathbbm{1}\{P_j\leq \min(k \alpha  \gamma_j/\hat{\pi}_0^t, \lambda) \} \geq k\right\},
\end{align}
\else 
\begin{align}
 &R_t^{\oSBH}\coloneqq \left\{i\leq t: P_i\leq {\min(k_t^* \alpha \gamma_i/\hat{\pi}_0^t, \lambda)}\right\}, \text{ where}
 \\
&k_t^* =\max\left\{k\in \{1,\ldots,t\}: \sum_{j=1}^t \mathbbm{1}\{P_j\leq \min(k \alpha  \gamma_j/\hat{\pi}_0^t, \lambda) \} \geq k\right\},
\end{align}
\fi
with the convention $\max(\emptyset)=0$. Note that for $\gamma_1=\ldots=\gamma_K=1/K$, we have $\hat{\pi}_0^K=\frac{1+\sum_{i=1}^K \mathbbm{1}\{P_i>\lambda\}}{(1-\lambda)K}$. In this case the online SBH procedure recovers the SBH procedure \citep{storey2002direct, storey2004strong} and therefore is a true generalization of it. In addition, since $\hat{\pi}_0^t$ is nonincreasing in $t$, the online SBH procedure is indeed an online ARC procedure.

 \begin{proposition}
     If the null p-values are independent from each other and the non-nulls, the online SBH procedure ensures that $\textnormal{OnlineFDR} \leq \alpha$.
 \end{proposition}

\begin{proof}
Let $\hat{\pi}_0^{t, -j}=\frac{\max_{i\in \mathbb{N}} \gamma_i +\sum_{i=1, i\neq j}^{\infty} \gamma_i \mathbbm{1}\{P_i>\lambda \lor i>t\}}{(1-\lambda)}$. For any $t\in \mathbb{N}$, we have 
\begin{align}
    \text{FDR}_t&=\sum_{i\leq t, i\in I_0} \mathbb{E}\left[ \frac{\mathbbm{1}\{P_i\leq {\min(k_t^* \alpha \gamma_i/\hat{\pi}_0^t, \lambda)}\}}{k_t^*} \right]\\ &
    {=}\sum_{i\leq t, i\in I_0} \mathbb{E}\left[ \frac{\mathbbm{1}\{P_i\leq {\min(k_t^* \alpha \gamma_i/\hat{\pi}_0^{t,-i}, \lambda)}\}}{k_t^*} \right]\\
    &\leq \sum_{i\leq t, i\in I_0} \mathbb{E}\left[ \frac{\mathbbm{1}\{P_i\leq {k_t^* \alpha \gamma_i/\hat{\pi}_0^{t,-i}}\}}{k_t^*} \right] \\
    &=\sum_{i\leq t, i\in I_0} \mathbb{E}\left[\mathbb{E}\left( \frac{\mathbbm{1}\{P_i\leq {k_t^* \alpha \gamma_i/\hat{\pi}_0^{t,-i}}\}}{k_t^*} \Bigg\vert (P_j)_{j\in \{1,\ldots,t\}\setminus i} \right)\right] \\
    &\stackrel{(i)}{\leq}  \sum_{i\leq t, i\in I_0} \alpha \gamma_i\mathbb{E} \left[ \frac{1}{\hat{\pi}_0^{t,-i}}\right] \\
    &\stackrel{(ii)}{\leq} \sum_{i\leq t, i\in I_0} \alpha \gamma_i \frac{1}{\max_{j\leq t} \gamma_j+\sum_{j\leq t, j\in I_0, j\neq i} \gamma_j}  \\
    &= \alpha\sum_{i\leq t, i\in I_0}   \frac{\gamma_i}{\max_{j\leq t} \gamma_j+\sum_{j\leq t, j\in I_0, j\neq i} \gamma_j} \leq \alpha.
\end{align}
    Inequality $(i)$ follows by Lemma 1 (a) and inequality $(ii)$ by Lemma 3 of \citet{ramdas2019unified}.
\end{proof}

An existing online procedure that adapts to the proportion of true hypotheses is SAFFRON \citep{ramdas2018saffron}, which is defined as LORD (see Section~\ref{sec:LORD}) but with the condition 
\ifarxiv
\begin{align}
    \frac{\sum_{i=1}^t \alpha_i^{\SAFFRON} \frac{\mathbbm{1}\{P_i>\lambda\}}{1-\lambda} }{|R_t^{\SAFFRON}|\lor 1}\leq \alpha \ \text{for all } t\in \mathbb{N}, \text{ where } R_t^{\SAFFRON}\coloneqq\{i\leq t: P_i\leq \alpha_i^{\SAFFRON}\} \label{eq:cond_saffron}
\end{align}
\else 
\begin{align}
    &\frac{\sum_{i=1}^t \alpha_i^{\SAFFRON} \frac{\mathbbm{1}\{P_i>\lambda\}}{1-\lambda} }{|R_t^{\SAFFRON}|\lor 1}\leq \alpha \ \text{for all } t\in \mathbb{N}, \label{eq:cond_saffron} \\ \text{where } &R_t^{\SAFFRON}\coloneqq\{i\leq t: P_i\leq \alpha_i^{\SAFFRON}\} 
\end{align}
\fi
instead of \eqref{eq:cond_lord}. The relation between SAFFRON and online SBH is similar to the relation between LORD and online BH (see Section~\ref{sec:LORD}). That means, the levels defined by online SBH satisfy \eqref{eq:cond_saffron}. However, in contrast to the LORD vs. online BH case, online SBH satisfies the SAFFRON condition \eqref{eq:cond_saffron} only in a conservative manner due to the additional summand $\max_{i\in \mathbb{N}} \gamma_i$ in \eqref{eq:estimate_prop}. We think this is required, since SAFFRON only allows to use the information $P_i>\lambda_i$ for all levels $\alpha_j^{\SAFFRON}$, $j>i$, while online SBH uses it for all $j\in \mathbb{N}$. Consequently, online SBH can lead to a more balanced weighting.  Overall, we think both procedures lead to a comparable power and the main difference lies in a different weighting of the hypotheses, just as for online BH and LORD. 

This can also be seen in Figure~\ref{fig:sim_SASFFRON}, where we compare online SBH with SAFFRON as implemented in the \texttt{R} package \texttt{onlineFDR} \citep{robertson2019onlinefdr}, which should be compared to Figure~\ref{fig:sim_LORD}. Both procedures were applied with $\lambda=0.5$, the remaining simulation setup is the same as in Section~\ref{sec:LORD}. 

\ifarxiv
\begin{figure*}[h!]
\centering
\includegraphics[width=0.8\textwidth]{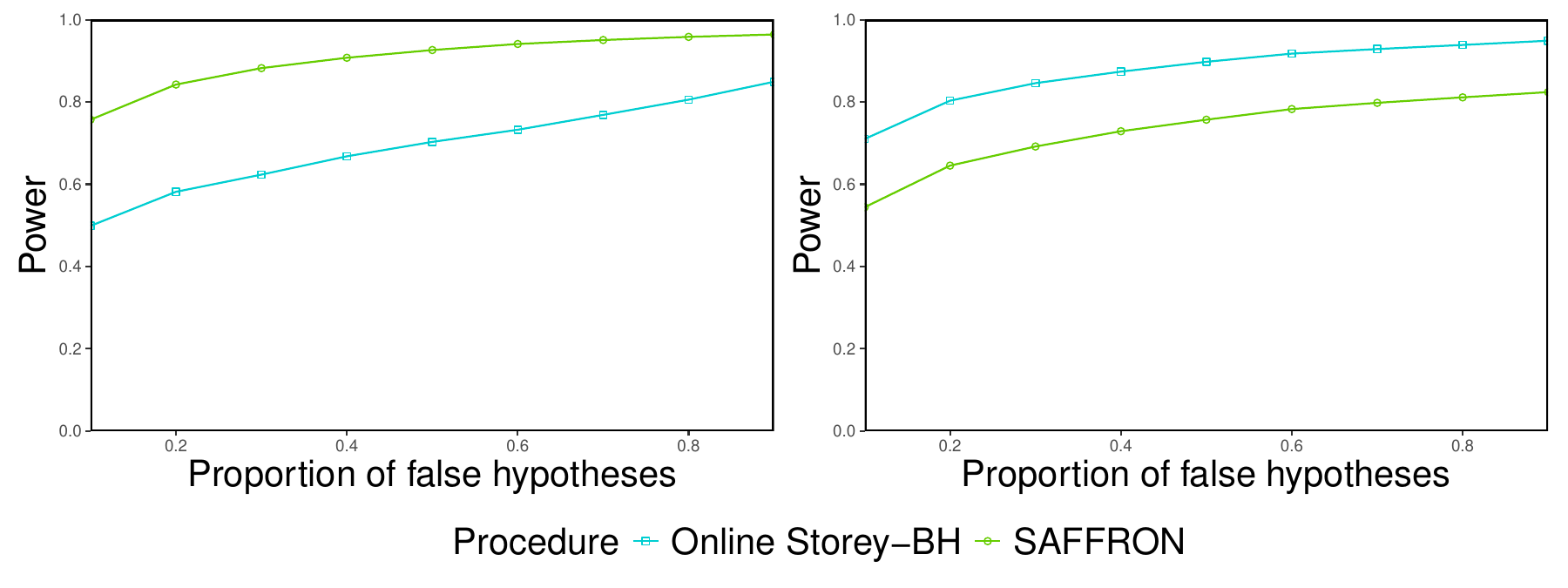}
\caption{Power comparison of online SBH and SAFFRON for different proportions of false hypotheses. In the left plot the sequence $(\gamma_t)_{t\in \mathbb{N}}$ decreases fast ($q=0.99$) and in the left plot it decreases slow ($q=0.999$). The simulation setup is described in Section \ref{appn:SBH}. \label{fig:sim_SASFFRON} }\end{figure*}
\else 
\begin{figure*}[h!]
\centering
\includegraphics[width=0.7\textwidth]{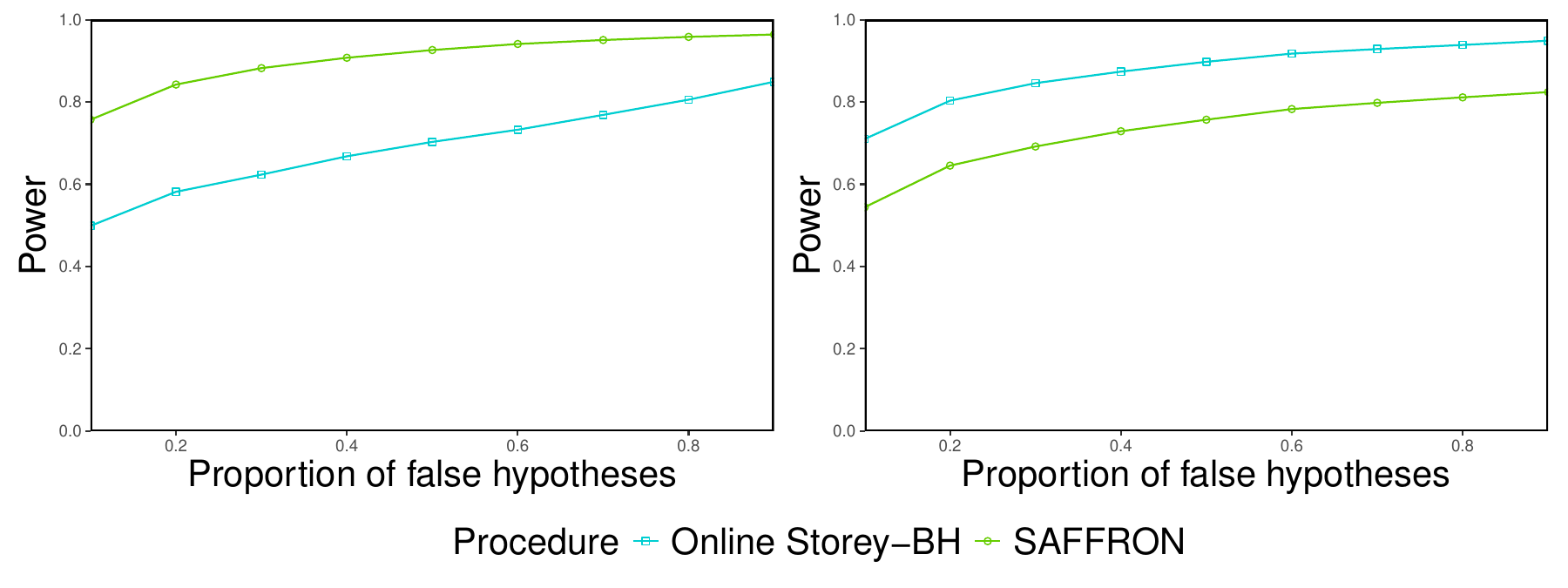}
\caption{Power comparison of online SBH and SAFFRON for different proportions of false hypotheses. In the left plot the sequence $(\gamma_t)_{t\in \mathbb{N}}$ decreases fast ($q=0.99$) and in the left plot it decreases slow ($q=0.999$). The simulation setup is described in Section \ref{appn:SBH}. \label{fig:sim_SASFFRON} }\end{figure*}
\fi

\section{Positive and negative dependence\label{sec:PRDS}}

Here we give the formal definitions of the used positive and negative dependence assumptions made in the paper, which are online extensions of the classical definitions \citep{benjamini2001control, su_fdr-linking_theorem_2018, chi_multiple_testing_2024, wang2022false}.

Let us first define two different positive dependence assumptions that one can make on p-values.
The first is the PRDS condition for FDR control of the BH procedure that was introduced by \citet{benjamini2001control}.
%
%
We generalize this typical PRDS definition to the online setting (i.e., applicable to countably infinite hypotheses) in the following way.
\begin{definition}[PRDS]\label{def:PRDS}
    A set $D \subseteq \reals^n$ is called \emph{increasing} if $\mathbf{x} \in D$ implies that $\mathbf{y} \in D$ for all $\mathbf{y}\geq \mathbf{x}$. For a sequence of p-values $(P_t)_{t \in I}$ with a countable index set $I$, we say that they have \emph{positive regression dependence on a subset (PRDS)} if, for any finite $I' \subseteq I$, increasing set $D \subseteq \reals^{I'}$, and null index $i \in I_0$, we have that $\prob{(P_t)_{t \in I'} \in D \mid P_i \leq s}$ is increasing over $s \in (0, 1]$.
\end{definition}
Note that this definition is equivalent to the standard PRDS definition when $I$ is finite. Further, this is the PRDS definition implicitly used by \citet{zrnic2021asynchronous}. For example, PRDS is satisfied if the p-values corresponding to true hypotheses are independent from all the other p-values. However, it also holds under some kind of positive dependence \citep{benjamini2001control}.

A slightly weaker notion of positive dependence is PRDN \citep{su_fdr-linking_theorem_2018}, which is restricted to the dependency between p-values that correspond to true hypotheses.

\begin{definition}[PRDN]
    For a sequence of p-values $(P_t)_{t \in I}$ with a countable index set $I$, we say that they have \emph{positive regression dependence within nulls (PRDN)} if, for any finite subset of the null hypotheses $I' \subseteq I_0$, increasing set $D \subseteq \reals^{I'}$, and null index $i \in I_0$, we have that $\prob{(P_t)_{t \in I'} \in D \mid P_i \leq s}$ is a increasing over $s \in (0, 1]$.
\end{definition}
Likewise, this is the generalization of PRDN from \citet{su_fdr-linking_theorem_2018} to the online setting.

The only type of negative dependence we consider in this paper is WNDN.

\begin{definition}[WNDN]
    We say that a sequence of p-values $(P_t)_{t \in I}$ are \emph{weakly negatively dependent within nulls (WNDN)} if $\prob{\bigcap_{i \in A} P_i \leq s} \leq \prod_{i \in A} \prob{P_i \leq s}$ for any $A \subseteq I_0$ and $s \in [0, 1]$.
\end{definition}

Similarly as for positive dependence, there are many types of negative dependence. A comparison of the different types is given by \citet{chi_multiple_testing_2024}.

Since e-values and p-values are working on inverse scales, the PRDS condition is also flipped around when considering e-values \citep{wang2022false}.
\begin{definition}[PRDS for e-values]\label{def:PRDS_e-values}
    A set $D \subseteq \reals^n$ is called \emph{decreasing} if $\mathbf{x} \in D$ implies that $\mathbf{y} \in D$ for all $\mathbf{y}\leq \mathbf{x}$. For a sequence of e-values $(E_t)_{t \in I}$ with a countable index set $I$, we say that they have \emph{positive regression dependence on a subset (PRDS)} if, for any finite $I' \subseteq I$, decreasing set $D \subseteq \reals^{I'}$, and null index $i \in I_0$, we have that $\prob{(E_t)_{t \in I'} \in D \mid E_i \geq s}$ is decreasing over $s \in (0, \infty)$.
\end{definition}
This definition of PRDS for e-values ensures that if the p-values $(P_t)_{t\in \mathbb{N}}$ are PRDS and $(\psi_t)_{t\in \mathbb{N}}$ are nonincreasing transforms, then the random variables $(\psi_t(P_t))_{t\in \mathbb{N}}$ are PRDS in the e-value sense.

\section{Simulation setup\label{sec:sim_setup}}

We simulated $m=100$ independent trials, each consisting of $n=1000$ hypothesis pairs of the form $H_t:X_t\sim N(0,1)$ vs. $H_t^A:X_t\sim N(\mu_A,1)$, $t\in \mathbb{N}$, where $\mu_A>0$ is the signal strength. In case of $\mu_A=3.5$ we speak of weak signal and in case of $\mu_A=4.5$ of strong signal.
To obtain locally dependent e-values, we first generated normally distributed random variables $(Z_1,\ldots,Z_n)$ with mean $0$, variance $1$ and a batch dependence structure. That means we have batches $B_i=\{Z_{(i-1)b+1},\ldots, Z_{ib}\}$, $i\in \{1,\ldots, n/b\}$, such that random variables contained in the same batch have correlation $0.5$ while random variables from different batches are independent. The final test statistics were obtained by $X_t=Z_t+\mu_A \Pi_A$, where $\Pi_A\sim \text{Ber}(\pi_A)$ and $\pi_A\in \{0.1, 0.2,\ldots, 0.9\}$ gives the probability that $H_t^A$ is true. 

The e-values were calculated as likelihood ratio $\phi_{\mu_A}(X_t)/\phi_{0}(X_t)$, where $\phi_{\mu}$ denotes the density of a normal distribution with variance $1$ and mean $\mu$. The boosted e-values were obtained as described in Examples \ref{example:boosting} and \ref{example:boosting_local} using the truncation functions $\leftindex^-{T}_t^{s}$ and $\leftindex^-{T}_t^{L_t, s}$ with $s=n$, respectively. In all simulations concerning e-values we set the batch-size to $b=20$ and applied all procedures with $\gamma_t=q^{t-1}(1-q)$, $t\in \{1,\ldots,n\}$, where $q=0.99$. 

 The p-values were calculated by $P_t=\Phi(-Z_t)$. When considering p-values, we generated independent test statistics ($b=1$) and only considered weak signals ($\mu_A=3.5$).

The code for the simulations is available at \url{https://github.com/fischer23/online_e-BH}.

\section{Omitted proofs and derivations\label{sec:omitted_proofs}}

In this section, we provide the omitted proofs or derivations for results that were mentioned in the main paper.

\subsection{Online e-BH with p-values\label{appn:proof_p-vals}}
In this subsection, we provide the formal proof of Proposition~\ref{prop:p-values}.
\begin{proof}
    If we prove that $\mathbb{E}[T_t(\psi_t(U))]\leq 1$, where $T_t$ is given by \eqref{eq:truncation_online} and $U\sim U[0,1]$, then the first claim follows by Proposition \ref{prop:boosted}. For this, note that
    \ifarxiv
    \begin{align}
        \mathbb{E}[T_t(\psi_t(U))]&=\sum_{k=1}^\infty \mathbb{P}\left( \frac{1}{k\alpha\gamma_t} \leq \psi_t(U) < \frac{1}{(k-1)\alpha\gamma_t}\right) \frac{1}{k\alpha \gamma_t}\\
        &= \sum_{k=1}^\infty \mathbb{P}\left( \psi^{-1}\left(\frac{1}{k\alpha\gamma_t}\right) \geq U > \psi^{-1}\left(\frac{1}{(k-1)\alpha\gamma_t}\right) \right) \frac{1}{k\alpha \gamma_t} \\
        &= \sum_{k=1}^\infty \frac{1}{k\alpha \gamma_t}\left(\psi_t^{-1}\left(\frac{1}{k \alpha \gamma_t}\right)- \psi_t^{-1}\left(\frac{1}{(k-1) \alpha \gamma_t}\right)\right) \\
        &\leq 1.
    \end{align}
    \else 
\begin{align}
        &\mathbb{E}[T_t(\psi_t(U))] \\ &=\sum_{k=1}^\infty \mathbb{P}\left( \frac{1}{k\alpha\gamma_t} \leq \psi_t(U) < \frac{1}{(k-1)\alpha\gamma_t}\right) \frac{1}{k\alpha \gamma_t}\\
        &= \sum_{k=1}^\infty \mathbb{P}\left( \psi^{-1}\left(\frac{1}{k\alpha\gamma_t}\right) \geq U > \psi^{-1}\left(\frac{1}{(k-1)\alpha\gamma_t}\right) \right) \frac{1}{k\alpha \gamma_t} \\
        &= \sum_{k=1}^\infty \frac{1}{k\alpha \gamma_t}\left(\psi_t^{-1}\left(\frac{1}{k \alpha \gamma_t}\right)- \psi_t^{-1}\left(\frac{1}{(k-1) \alpha \gamma_t}\right)\right) \\
        &\leq 1.
    \end{align}
    \fi
    The second claim follows immediately by Proposition~\ref{prop:boost_PRDS}, since
    \ifarxiv
    \begin{align}
    \sup_{k\in \mathbb{N}} \frac{1}{k \alpha \gamma_t} \mathbb{P}(\psi_t(U)\geq 1/(k\alpha \gamma_t))&=\sup_{k\in \mathbb{N}} \frac{1}{k \alpha \gamma_t} \mathbb{P}(U\leq \psi_t^{-1}(1/(k\alpha \gamma_t)))\\ &= \sup_{k\in \mathbb{N}} \frac{1}{k \alpha \gamma_t} \psi_t^{-1}(1/(k\alpha \gamma_t)).
    \end{align}
    \else 
\begin{align}
    &\sup_{k\in \mathbb{N}} \frac{1}{k \alpha \gamma_t} \mathbb{P}(\psi_t(U)\geq 1/(k\alpha \gamma_t)) \\ &=\sup_{k\in \mathbb{N}} \frac{1}{k \alpha \gamma_t} \mathbb{P}(U\leq \psi_t^{-1}(1/(k\alpha \gamma_t)))\\ &= \sup_{k\in \mathbb{N}} \frac{1}{k \alpha \gamma_t} \psi_t^{-1}(1/(k\alpha \gamma_t)).
    \end{align}
    \fi
\end{proof}

\subsection{SupFDR control of\ online BR and r-LOND\label{sec:appn_supfdr}}
%
%
%

In the following proposition, we show how the online BR procedure can be written as a special instance of the online e-BH method and therefore provides SupFDR control under arbitrary dependence between the e-values.

\begin{proposition}\label{prop:BR_procedure}
    The online BR procedure provides SupFDR control at level $\alpha$ under arbitrary dependence between the p-values.
\end{proposition}
\begin{proof}
    Let $\beta_t$ be a shape function and $\psi_t$ be chosen such that $\psi_t^{-1}(1/(\alpha \gamma_t k))=\alpha \beta_t(k) \gamma_t$. Note that $\psi_t(P_t)\geq 1/(\alpha \gamma_t k)$ iff $P_t\leq \psi_t^{-1}(1/(\alpha \gamma_t k))=\alpha \beta_t(k) \gamma_t$. Hence, online e-BH applied to $\psi_1(P_1),\psi_2(P_2),\ldots $ rejects the same hypotheses as online BR applied to $P_1, P_2,\ldots$ . Furthermore, it holds that
    \ifarxiv
\begin{align}
    \sum_{k=1}^\infty \frac{1}{k\alpha \gamma_t}\left(\psi_t^{-1}\left(\frac{1}{k \alpha \gamma_t}\right)- \psi_t^{-1}\left(\frac{1}{(k-1) \alpha \gamma_t}\right)\right)&=\sum_{k=1}^\infty \frac{1}{k}\left(\beta_t(k)- \beta_t(k-1)\right)\\
    &=\sum_{k=1}^\infty \int_{k-1}^k x/k \ d\nu(x) \\
    &\leq \int_0^\infty 1 \ d\nu(x) \\
    &=1.
\end{align}
\else 
\begin{align}
    &\sum_{k=1}^\infty \frac{1}{k\alpha \gamma_t}\left(\psi_t^{-1}\left(\frac{1}{k \alpha \gamma_t}\right)- \psi_t^{-1}\left(\frac{1}{(k-1) \alpha \gamma_t}\right)\right) \\ &=\sum_{k=1}^\infty \frac{1}{k}\left(\beta_t(k)- \beta_t(k-1)\right)\\
    &=\sum_{k=1}^\infty \int_{k-1}^k x/k \ d\nu(x) \\
    &\leq \int_0^\infty 1 \ d\nu(x) \\
    &=1.
\end{align}
\fi 
Consequently, the SupFDR control follows by Proposition \ref{prop:p-values}.
\end{proof}

In the same manner as we derived the online BR procedure as special case of the boosted online e-BH procedure, one could also derive r-LOND as a special case of a boosted e-LOND procedure. For this, one can just define the truncation function
\ifarxiv
    $$
    T_t^{\text{LOND}}=\sum_{k=1}^t \mathbbm{1}\left\{ \frac{1}{k\alpha\gamma_t} \leq x < \frac{1}{(k-1)\alpha\gamma_t}\right\} \frac{1}{k\alpha \gamma_t} \quad \text{with } T_t(\infty)=\frac{1}{\alpha\gamma_t}
    $$
    \else 
$$
    T_t^{\text{LOND}}=\sum_{k=1}^t \mathbbm{1}\left\{ \frac{1}{k\alpha\gamma_t} \leq x < \frac{1}{(k-1)\alpha\gamma_t}\right\} \frac{1}{k\alpha \gamma_t}
    $$
    with $T_t(\infty)=\frac{1}{\alpha\gamma_t}$
    \fi 
    and then do the same steps as above. Consequently, r-LOND also provides SupFDR control.

\begin{proposition}\label{prop:r-LOND}
    The r-LOND procedure provides SupFDR control at level $\alpha$ under arbitrary dependence between the p-values.
\end{proposition}

\subsection{Weighted \texttt{FDR-Linking} theorem\label{sec:proofs}}
Here we provide the proof of Theorem~\ref{theo:FDR-linking}.
\begin{proof}
We follow a similar proof structure as the proof of Theorem 1 in \citet{su_fdr-linking_theorem_2018}. 
By the definition of self-consistency, we note that any $\rejset \in \mathcal{R}(\alpha)$ and $t \in \rejset$, the following inequality holds:
\begin{align}
        |\rejset| \geq \lceil (\alpha\gamma_t)^{-1}  P_t\rceil. \label{eq:rejset-lb}
\end{align}
Thus, we get an upper bound on the supremum over FDP
\ifarxiv
\begin{align}
    \sup_{\rejset \in \Rcal(\alpha)}\ \FDP(R)
    &=\max_{\rejset \in \Rcal(\alpha)} \frac{|I_0 \cap R|}{|R| \vee 1}
    \leq \max_{\rejset \in \Rcal(\alpha)} \frac{|I_0 \cap R|}{\max\{\lceil (\alpha \gamma_t)^{-1}  P_t\rceil : t \in \rejset\} \cup \{1\}}.
\end{align}
\else 
\begin{align}
    \sup_{\rejset \in \Rcal(\alpha)}\ \FDP(R)
    &=\max_{\rejset \in \Rcal(\alpha)} \frac{|I_0 \cap R|}{|R| \vee 1}
    \\ &\leq \max_{\rejset \in \Rcal(\alpha)} \frac{|I_0 \cap R|}{\max\{\lceil (\alpha \gamma_t)^{-1}  P_t\rceil : t \in \rejset\} \cup \{1\}}.
\end{align}
\fi
The inequality is a result of applying the bound in \eqref{eq:rejset-lb} to every $t \in R$ and taking the best possible bound.

Let $K_0 \coloneqq |I_0|$ and $P_1^\tnull, \dots, P_{K_0}^\tnull$ and $\gamma_1^\tnull, \dots, \gamma_{K_0}^\tnull$ denote the corresponding weights. Now, define $P_{(1)}^\tnull, \dots, P_{(K_0)}^\tnull$ and $\gamma_{(1)}^\tnull, \dots, \gamma_{(K_0)}^\tnull$ be the null p-values and weights sorted in ascending order based on $(\gamma_i^\tnull)^{-1} P_i^\tnull$. We continue the derivation of the upper bound as follows.
\ifarxiv
\begin{align}
    \sup_{\rejset \in \Rcal(\alpha)}\ \FDP(R)
    & \leq \max_{\rejset \in \Rcal(\alpha)} \frac{|I_0 \cap R|}{\lceil (\alpha \gamma_{(|I_0 \cap R|)}^\tnull)^{-1} P_{(|I_0 \cap R|)}^\tnull\rceil}
     \leq \max_{j \in [K_0]} \frac{j}{\lceil (\alpha \gamma_{(j)}^\tnull)^{-1} P_{(j)}^\tnull\rceil}\\
    &\leq  \frac{\alpha}{\min_{j \in [K_0]} (j \gamma_{(j)}^\tnull)^{-1} P_{(j)}^\tnull} = \frac{\pi_0 \alpha}{\min_{j \in [K_0]} (j \gamma_{(j)}^\tnull / \pi_0)^{-1} P_{(j)}^\tnull}. \label{eq:supfdp-ub}
\end{align}
\else 
\begin{align}
    \sup_{\rejset \in \Rcal(\alpha)}\ \FDP(R)
    & \leq \max_{\rejset \in \Rcal(\alpha)} \frac{|I_0 \cap R|}{\lceil (\alpha \gamma_{(|I_0 \cap R|)}^\tnull)^{-1} P_{(|I_0 \cap R|)}^\tnull\rceil}
     \\ &\leq \max_{j \in [K_0]} \frac{j}{\lceil (\alpha \gamma_{(j)}^\tnull)^{-1} P_{(j)}^\tnull\rceil}\\
    &\leq  \frac{\alpha}{\min_{j \in [K_0]} (j \gamma_{(j)}^\tnull)^{-1} P_{(j)}^\tnull}\\ &= \frac{\pi_0 \alpha}{\min_{j \in [K_0]} (j \gamma_{(j)}^\tnull / \pi_0)^{-1} P_{(j)}^\tnull}. \label{eq:supfdp-ub}
\end{align}
\fi 
The first inequality is noting that there exists a null p-value and corresponding weight of index $i \in R$ that satisfies $(\gamma_{(i)}^\tnull)^{-1}P_{(i)}^\tnull \leq (\gamma_{(|I_0 \cap \rejset|)}^\tnull)^{-1}P_{(|I_0 \cap \rejset|)}^\tnull$ simply by the cardinality of $I_0 \cap \rejset$. The second inequality is by changing the indexing set of the maximum, and the third inequality is by dropping the ceiling function in the denominator.

Now we note that
\begin{align}
    P^\nullSimes \coloneqq \min_{j \in [K_0]} (j \gamma_{(j)}^\tnull / \pi_0)^{-1} P_{(j)}^\tnull
\end{align} is precisely the weighted Simes p-value applied only to the p-values and weights of the null hypotheses, i.e., $I_0$. One rejects the weighted Simes p-value if and only if the weighted BH procedure makes any discoveries at level $\alpha \in [0, 1]$.

We now observe that trivially, we can augment the bound in \eqref{eq:supfdp-ub} to get that
\begin{align}
     \sup_{\rejset \in \Rcal(\alpha)}\ \FDP(R)  \leq \left(\frac{\pi_0 \alpha}{P^{\nullSimes}}\right) \wedge 1. \label{eq:min-supfdp-bound}
\end{align}

Now, let $F$ denote the c.d.f. of $P^{\nullSimes}$.
Thus, we get the following derivation:
\ifarxiv
\begin{align}
    \expect\left[\sup_{\rejset \in \Rcal(\alpha)}\ \FDP(R) \right]
    &\leq \expect\left[\left(\frac{\pi_0 \alpha}{P^{\nullSimes}}\right) \wedge 1\right]\\
    &=\prob{P^\nullSimes \leq \pi_0\alpha} + \expect\left[\frac{\pi_0\alpha}{P^\nullSimes} \cdot \ind{P^\nullSimes > \pi_0\alpha}\right]\\
    &= F(\pi_0\alpha) + \int\limits_{\pi_0\alpha}^1\frac{\pi_0\alpha}{x}\ \nd F(x).
\end{align}
\else
\begin{align}
    &\expect\left[\sup_{\rejset \in \Rcal(\alpha)}\ \FDP(R) \right]
    \\ &\leq \expect\left[\left(\frac{\pi_0 \alpha}{P^{\nullSimes}}\right) \wedge 1\right]\\
    &=\prob{P^\nullSimes \leq \pi_0\alpha} \\ &+ \expect\left[\frac{\pi_0\alpha}{P^\nullSimes} \cdot \ind{P^\nullSimes > \pi_0\alpha}\right]\\
    &= F(\pi_0\alpha) + \int\limits_{\pi_0\alpha}^1\frac{\pi_0\alpha}{x}\ \nd F(x).
\end{align}
\fi 
The first inequality is by plugging in \eqref{eq:min-supfdp-bound}. The second equality follows from casing on the value of $P^\nullSimes$ into disjoint cases, and the last equality comes from the definition of $F$.


Using integration by parts, it follows that

\ifarxiv
\begin{align}
    \expect\left[\sup_{\rejset \in \Rcal(\alpha)}\ \FDP(R) \right]&\leq F(\pi_0\alpha) + \pi_0\alpha\left(F(1) - \frac{F(\pi_0\alpha)}{\pi_0\alpha} + \int\limits_{\pi_0\alpha}^1\frac{F(x)}{x^2}\ \nd x\right)\\
    &=\pi_0\alpha + \pi_0\alpha \int\limits_{\pi\alpha}^1 \frac{F(x)}{x^2}\ \nd x
    =\pi_0\alpha + \pi_0\alpha \int\limits_{\pi\alpha}^1 \frac{\FDR_0(x)}{x^2}\ \nd x.
\end{align} 
\else 
\begin{align}
    &\expect\left[\sup_{\rejset \in \Rcal(\alpha)}\ \FDP(R) \right]\\ &\leq F(\pi_0\alpha) + \pi_0\alpha\left(F(1) - \frac{F(\pi_0\alpha)}{\pi_0\alpha} + \int\limits_{\pi_0\alpha}^1\frac{F(x)}{x^2}\ \nd x\right)\\
    &=\pi_0\alpha + \pi_0\alpha \int\limits_{\pi\alpha}^1 \frac{F(x)}{x^2}\ \nd x
    =\pi_0\alpha + \pi_0\alpha \int\limits_{\pi\alpha}^1 \frac{\FDR_0(x)}{x^2}\ \nd x.
\end{align} 
\fi
The last equality follows from the equivalence between Type I error of weighted Simes p-value and FDR of weighted BH under the global null (i.e., when applied to only null p-values). Hence we have shown our desired result.
\end{proof}

\subsection{Relationship between SupFDR and StopFDR\label{sec:appn_supFDR_FDR}}

We noted in the introduction that a bound on the FDR at all stopping times also implies a bound on the SupFDR. In this section we give a formal derivation of this result.

\begin{definition}    Recall that a filtration $(\Fcal_t)_{t \in \naturals}$ is a sequence of nested sigma-algebras, and stopping times (w.r.t.\ $(\Fcal_t)$) are random variables $\tau \in \naturals$ such that $\ind{\tau = t}$ is measurable w.r.t.\ $\Fcal_t$. 
\end{definition}

\begin{theorem}
    Let $(W_t)_{t \in \naturals}$ be a nonnegative process that is adapted to some filtration $(\Fcal_t)$, meaning $W_t$ is measurable with respect to $\mathcal{F}_t$ for all $t\in \mathbb{N}$. Furthermore, let $(W_t)_{t \in \naturals}$ be almost surely bounded (with respect to $\mathbb{P}$) by some constant $c \geq 1$ and $\mathbb{E}[W_{\tau}]\leq 1$ for all stopping times $\tau$ with respect to $(\Fcal_t)$. Then, $\expect[\sup_{t \in \naturals} W_t] \leq 1 + \log(c)$.
\end{theorem}
\begin{proof}
    By Ville's inequality, we have that $$\prob{\sup_{t \in \naturals} W_t > \alpha^{-1}} \leq \alpha$$ for all $\alpha \in [0, 1]$. Thus, we can use the tail bound integration formula for expectation of nonnegative random variables:
    \ifarxiv
    \begin{align}
        \expect\left[\sup_{t \in \naturals} W_t\right]
        &\leq \int\limits_0^\infty \prob{\sup_{t \in \naturals} W_t \geq s}\ ds
        \\&= \int\limits_0^c \prob{\sup_{t \in \naturals} W_t \geq s}\ ds
        \leq\int\limits_0^c s^{-1} \wedge 1\ ds \\
        &= \int\limits_0^1\ ds + \int\limits_1^c s^{-1}\ ds
        = 1 + \log(c)
    \end{align}
    \else  
\begin{align}
        \expect\left[\sup_{t \in \naturals} W_t\right]
        &\leq \int\limits_0^\infty \prob{\sup_{t \in \naturals} W_t \geq s}\ ds
        \\&= \int\limits_0^c \prob{\sup_{t \in \naturals} W_t \geq s}\ ds
        \leq\int\limits_0^c s^{-1} \wedge 1\ ds \\
        &= \int\limits_0^1\ ds + \int\limits_1^c s^{-1}\ ds
        = 1 + \log(c)
    \end{align}
    \fi
\end{proof}

\begin{corol}
If a sequence of discovery sets $(R_t)$ ensures StopFDR control at level $\alpha$, then SupFDR is controlled at level $\alpha(1 + \log(\alpha^{-1}))$.
\end{corol}

The above theorem can also be seen as similar to defining an adjuster \citep{shafer_test_martingales_2011,dawid_insuring_loss_2011,dawid_probability-free_pricing_2011,choe_combining_evidence_2024} for upper bounded e-processes or calibrator \citep{vovk_logic_probability_1993,vovk_e-values_calibration_2021} for lower bounded anytime-valid p-values.

\subsection{Sharpness of SupFDR and StopFDR control of online BH\label{sec:appn_obh_sharp}}

We will show that the SupFDR and StopFDR control of online BH under PRDN p-values in \Cref{thm:onlinebh-supfdr} is sharp, as claimed in Theorem~\ref{thm:sharp-supfdr}. First, we cite the following construction from \citet{su_fdr-linking_theorem_2018} for a worst-case construction of the FDP of picking non-nulls as a function of the realized null p-values. Let there be $K$ fixed hypotheses, and without loss of generality let $I_0 = [|K_0|]$, i.e., the nulls come first. Now, define the following quantities:
\begin{align}
    j^* \coloneqq \underset{j \in I_0}{\text{argmax}}\ j / \lceil K P^{\text{null}}_{(j)} / \alpha\rceil,\\
    K_1^* \coloneqq (\lceil K P_{(j^*)}^\textnormal{null} / \alpha \rceil - j^*) \vee 0,
\end{align} where $P^{\text{null}}_{(1)}, \dots, P^{\text{null}}_{(K_0)}$ are the p-values corresponding to null hypotheses ordered from smallest to largest.
Consider the following adversarial construction of non-null p-values for $t \in \{K_0 + 1, \dots, K\}$.
\begin{align}
    P_t  = 1 - \ind{t \leq K_0 + K_1^*}. \label{eq:adv-pvalue}
\end{align}
\begin{fact}[Lemma 2.4 of \citet{su_fdr-linking_theorem_2018}]\label{fact:su-cex}
    Let $R^\BH$ be the discovery set of the BH procedure at level $\alpha$ applied to any p-values where the non-null p-values are set in the manner of \eqref{eq:adv-pvalue}. When $K_1^* \leq K - K_0$, then $\FDP(R^\BH) = \max_{j \in I_0} (j / \lceil Kp_{(j^*)}^{\textnormal{null}}/\alpha\rceil) \wedge 1$.
\end{fact}
\Cref{fig:lb-diagram} shows how we derive our problem instance for the online ARC setting from the aforementioned offline construction. Let $\gamma_t = K^{-1}$ for $t \in [K]$ and $\gamma_t = 0$ for $t > K$. We let the first $K_0$ p-values be identical to the offline setting and the remaining hypotheses be non-null with p-values set to 0. Now we note the following.
\begin{lemma}\label{lemma:fdp-equiv}
    For the above choice of p-values and weights for the online setting, we have that $R^{\oBH}_{K_0 + K_1^*} = R^\BH$ and as a result, we get the following identity for the FDP when $K_1^* \leq K - K_0$: $\FDP\left(R^{\oBH}_{K_0 + K_1^*}\right) = (\max_{j \in I_0} j / \lceil Kp_{(j^*)}^{\textnormal{null}} / \alpha \rceil) \wedge 1.$
\end{lemma}
\begin{proof}    
Note that we can map p-value realizations of our offline construction to our online construction simply by setting the null p-values to be equivalent but setting all non-null p-values to be 0 in the online construction. Applying weighted BH with weights of $\gamma_t = K^{-1}$ to the first $K_0 + K_1^*$ hypotheses is exactly equivalent to applying standard BH to $K$ hypotheses where the first $K_0 + K_1^*$ hypotheses have identical p-values to the weighted BH p-values, and the remaining p-values are 1, as the hypotheses with p-values of 1 will never be rejected. Thus, when combined with Fact~\ref{fact:su-cex}, we can justify the above lemma.
\end{proof}
\ifarxiv
\begin{figure*}[h]
    \centering
    \includegraphics[width=\linewidth]{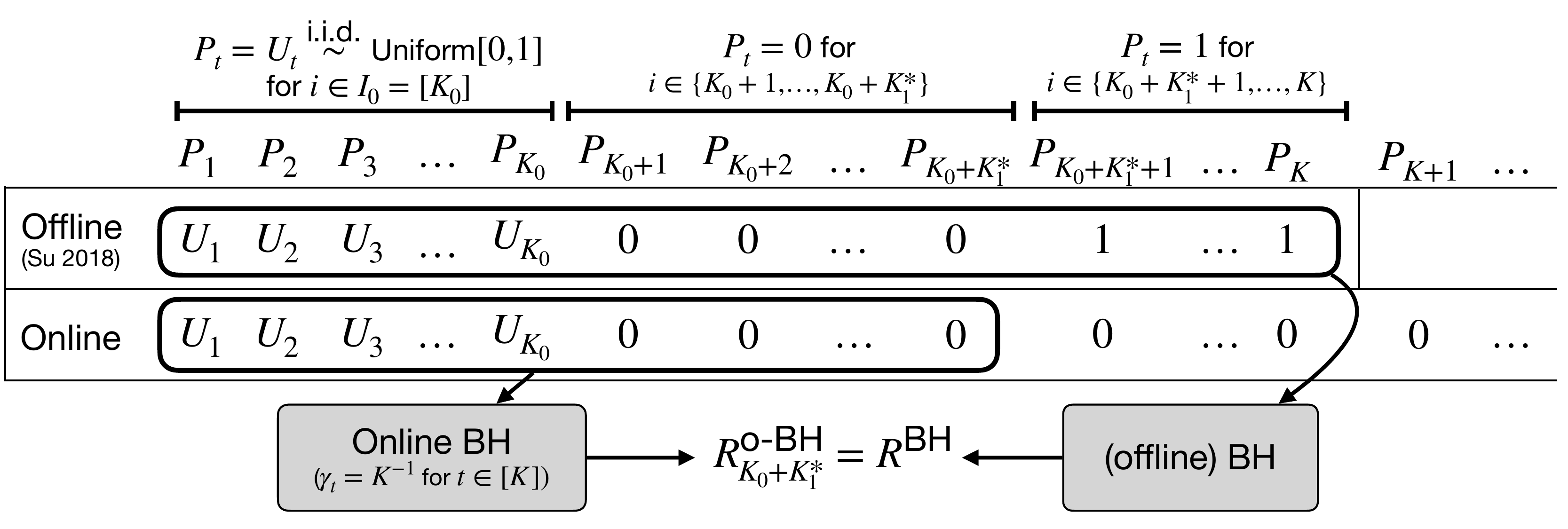}
    \caption{Visualization of our p-value construction for the online ARC setting that results in online BH producing discovery sets that are equal to the discovery sets in \citet{su_fdr-linking_theorem_2018} produced by the (offline) BH procedure (which maximize the FDR). This occurs whenever $K_1^* \leq K - K_0$, and we formalize the discovery equivalence result in \Cref{lemma:fdp-equiv}.}
    \label{fig:lb-diagram}
\end{figure*}
\else 
\begin{figure*}[h]
    \centering
    \includegraphics[width=0.7\linewidth]{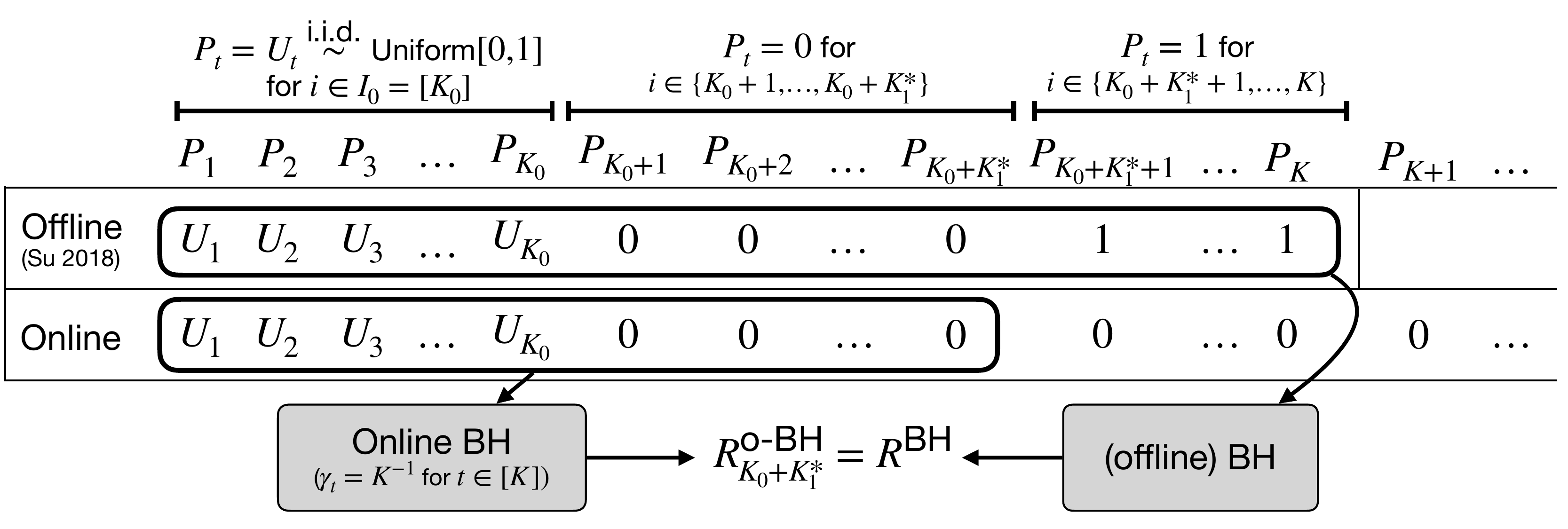}
    \caption{Visualization of our p-value construction for the online ARC setting that results in online BH producing discovery sets that are equal to the discovery sets in \citet{su_fdr-linking_theorem_2018} produced by the (offline) BH procedure (which maximize the FDR). This occurs whenever $K_1^* \leq K - K_0$, and we formalize the discovery equivalence result in \Cref{lemma:fdp-equiv}.}
    \label{fig:lb-diagram}
\end{figure*}
\fi

Now, \citet{su_fdr-linking_theorem_2018} showed that the FDP expression in the above lemma was sufficient to justify sharpness of the FDR bound when the null p-values are chosen to be i.i.d.\ uniform random variables.
\begin{fact}[Theorem 6 of \citet{su_fdr-linking_theorem_2018}]
    For every $\varepsilon > 0$, there exists a $K$, $K_0 < K$, and $\alpha' \in (0, 1)$ such that $\FDR(R^{\BH}) > (1 - \varepsilon)\alpha(1 + \log(\alpha^{-1}))$ for all $\alpha \leq \alpha'$ where the null p-values are i.i.d. uniform random variables, and the non-null p-values are set in accordance with \eqref{eq:adv-pvalue}.
\end{fact}

Now, we note that $K_0 + K_1^*$ is a stopping time since it can be determined solely by the null p-values (which we will have observed by the $K_0$th time step). Thus, if we apply our \Cref{lemma:fdp-equiv} to the above fact, we have shown our desired result in \Cref{thm:sharp-supfdr}.

As a result, SupFDR and StopFDR control of online BH is sharp in the PRDN case, and the $\log(\alpha^{-1})$ factor is unimprovable in the WNDN case.

\begin{remark}
We will also note what our example here means w.r.t.\ the sharpness of certain self-consistent procedures in the offline multiple testing setting.
\begin{itemize}
\item If we are thinking about the self-consistent set that has the maximum FDP over all self-consistent sets, then our construction does indicate that under i.i.d. null and non-null p-values, a FDR lower bound of $\alpha(1 +\log(\alpha^{-1}))$ is sharp.
\item However, we can also restrict ourselves to discovery sets along the ``BH path'', i.e., self-consistent sets that correspond to to the smallest $k$ p-values for some $k \in [K]$. The discovery sets used for our example would not be along the BH path, since the null p-values would not be 0 almost surely, and anything along the BH path would discover all the non-null p-values first (as they are all 0). In this case, we would still have i.i.d.\ null p-values, but some of the non-null p-values must be adversarially chosen in the sense of \citet{su_fdr-linking_theorem_2018}, i.e., setting the $K_1^*$ non-null p-values to be 1, to restrict the ``BH path'' to reach the sharp lower bound of $\alpha(1 + \log(\alpha^{-1}))$.
\end{itemize}
\end{remark}

\end{appendix}
\end{document}